%% file: main.tex
\documentclass[10pt,orivec]{llncs}

\usepackage[T1]{fontenc}
\usepackage{amsmath, amsfonts, amssymb}

\usepackage{amsthm}
\usepackage{physics2}
\usepackage{siunitx}

\AtBeginDocument{\RenewCommandCopy\qty\SI} 

\usepackage{xcolor} 
\usepackage{graphicx}
\graphicspath{{pdf/}}
\usepackage{standalone}

\usepackage{tikz}
\usetikzlibrary{fit, shapes.arrows, positioning, calc, shadows, backgrounds, arrows.meta, decorations.pathreplacing}
\usepackage{quantikz} 

\usepackage{pgfplots}
\usepgfplotslibrary{groupplots}
\pgfplotsset{compat=1.18}

\usepackage{booktabs}
\usepackage{multirow}
\usepackage{array}
\usepackage{diagbox}

\usepackage{subcaption}
\usepackage[ruled, vlined, linesnumbered]{algorithm2e}
\SetKw{KwDownTo}{down to}

\usepackage{pifont}   
\usepackage{cprotect}

\usepackage{longtable}
\usepackage{booktabs}
\usepackage{multirow}

\usepackage[colorlinks]{hyperref}
\usepackage[capitalize]{cleveref}

\newcommand{\cmark}{\textcolor{green!60!black}{\ding{51}}}
           
\newcommand{\dash}{\textemdash}                           

\newcommand{\stSafe}{\textcolor{green!60!black}{\textbf{True}}}
\newcommand{\stLogic}{\textcolor{red}{\texttt{LogicError}}}
\newcommand{\stPhase}{\textcolor{red}{\texttt{PhaseError}}}
\newcommand{\stHybrid}{\textcolor{red}{\texttt{BothError}}}

\newcommand{\stSep}{\textcolor{blue}{Separable}}
\newcommand{\stInsep}{\textcolor{red}{\textbf{Entanglement}}}

\newcommand{\outSafe}{\textcolor{green!60!black}{$\boldsymbol{U_{\mathrm{fix}}}$}}
\newcommand{\outUnsafe}{\textcolor{red}{$\boldsymbol{U_{\mathrm{fix}}}$}}
\newcommand{\A}{\mathcal{A}}
\newcommand{\W}{\mathcal{W}}

\usepackage{hyperref}
\hypersetup{
    colorlinks=true,         
    linkcolor=blue,          
    urlcolor=blue,           
    citecolor=blue           
}

\renewcommand{\orcidID}[1]{\smash{\href{http://orcid.org/#1}{\protect\raisebox{-1.25pt}{\protect\includegraphics{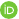}}}}}

\begin{document}

\title{Formal Verification of Quantum Ancilla Safety}

\author{
Jiqi Li\inst{1,2}\orcidID{0009-0009-6550-8087} \and
Jingyi Mei\inst{3}\orcidID{0000-0002-4665-9818} \and
Wang Fang\inst{4}\orcidID{0000-0001-7628-1185} \and
Ji Guan\inst{1}\thanks{Corresponding author.}\orcidID{0000-0002-3490-0029}
}

\institute{
Key Laboratory of System Software (Chinese Academy of Sciences), \\
Institute of Software, Chinese Academy of Sciences, Beijing, China \\
\email{\{lijiqi, guanj\}@ios.ac.cn}
\and
University of Chinese Academy of Sciences, Beijing, China
\and
Leiden University, The Netherlands \\
\email{j.mei@liacs.leidenuniv.nl}
\and
University of Edinburgh, United Kingdom \\
\email{wang.fang@ed.ac.uk}
}

\maketitle   

\pagenumbering{arabic}
\setcounter{page}{1}
\pagestyle{plain}
\begin{abstract}
Ensuring ancilla safety is a critical correctness requirement for quantum compilation, since ancilla qubits are routinely introduced to implement complex operations with fewer gates and reduced depth. However, formally verifying this property is computationally hard due to state-space explosion in the number of qubits, particularly for dirty ancillae, which carry unknown initial states and must be restored after use. We propose an end-to-end verification-and-repair framework that rigorously addresses both clean and dirty ancilla safety. Our core contribution is a two-step reduction strategy: we first prove that verifying an $m$-qubit dirty ancilla register decomposes into $2m$ independent clean ancilla safety checks; subsequently, we reduce each clean ancilla safety instance to an algebraic commutativity check against Pauli-$Z$ and Pauli-$X$ operators. This approach yields an efficient and naturally parallel verifier and enables actionable diagnosis by classifying violations into logic errors and phase errors. Leveraging this diagnosis, we further design lightweight repair routines that append local single-qubit rotations to eliminate a broad class of local ancilla faults. We implement the full pipeline in a prototype tool using a dual-backend architecture combining decision diagrams and weighted model counting, and validate it on diverse circuits ranging from arithmetic benchmarks to Grover’s algorithm. Our experiments demonstrate scalability to thousands of qubits and show that the proposed repairs effectively improve ancilla safety while preserving circuit functionality.
\end{abstract}

\section{Introduction}
In the current Noisy Intermediate-Scale Quantum (NISQ) era, qubit (quantum bit) scarcity remains a major bottleneck for achieving practical quantum advantage~\cite{preskill2025beyond}. 
To reduce space overhead, quantum compilers and circuit designers routinely introduce \emph{ancilla qubits} as temporary workspace, enabling more efficient decompositions of target unitaries and often reducing circuit depth under hardware constraints~\cite{barenco1995elementary,Itoko2020}. Ancillae (ancilla qubits) are broadly classified into two classes: \emph{clean} ancilla qubits and \emph{dirty} ancilla qubits. 
While \emph{clean} ancilla qubits initialized to $\ket{0}$ are the standard setting~\cite{shor1997}, \emph{dirty} ancilla qubits, qubits borrowed from idle parts of a larger computation, carrying unknown (and potentially entangled) states, can further improve space efficiency and flexibility~\cite{Low2024,Haner2017}.

However, exploiting dirty ancilla qubits introduces a stringent correctness requirement.
A circuit is \emph{dirty safe} only if it restores each dirty ancilla to \emph{exactly} its pre-use state and leaves no residual correlation with the working register, \emph{for all} initial system states (including those where the ancilla is initially entangled with the working qubits)~\cite{Su2025}.
This universal quantification makes dirty safety verification substantially harder than standard simulation~\cite{Su2025,Jones_2019}.

Existing approaches fall short of a scalable, general solution.
Simulation and randomized testing provide limited confidence only.
Prior formal methods often rely on sufficient-but-not-necessary structural templates (e.g., compute--uncompute patterns)~\cite{qwire,reqwire}, or depend on state representations with concrete algebraic weights that cannot natively capture the arbitrary, unknown initial states required for dirty ancilla verification~\cite{PVQC}.
Moreover, while dirty safety can be reduced to checking over a small set of basis states~\cite{Su2025}, the developed automated procedure targets classical reversible circuits and does not provide a general reduction principle (nor an end-to-end workflow with repair guidance) for universal quantum circuits.

Beyond verification, compilation pipelines also need \emph{actionable diagnostic feedback}. 
When ancilla safety fails, the failure may either return with the wrong computational-basis value or appear correct on basis states while still carrying a hidden phase error on superposition inputs~\cite{Su2025}, and these faults call for different fixes (such as adding a full uncomputation step~\cite{Paradis2022Reqomp,ParadisBSV2021uncompute} versus inserting a local phase). 
Without a principled diagnosis, developers are left with ad hoc debugging and expensive redesign iterations~\cite{Su2025,ying2024foundations}, and unsafe ancilla qubits cannot be reliably reused across optimizations~\cite{Itoko2020}, directly undermining the qubit-saving benefits that motivated the adoption of dirty ancilla qubits in the first place~\cite{barenco1995elementary,Low2024}.

We address this need~\cite{Su2025} with a fully automated backend-agnostic framework for \emph{verifying and repairing} dirty ancilla safety in unitary quantum circuits.
The core idea is a compact reduction chain that turns the semantic safety specification into efficiently checkable algebraic constraints~\cite{QCQI}.
First, we generalize clean safety to arbitrary target ancilla states and show that dirty safety of an $m$-qubit ancilla register decomposes into $2m$ independent clean-safety obligations on two fixed, linearly independent and non-orthogonal single-qubit states (equivalently, the Pauli-$Z$ and Pauli-$X$ eigenstates).
Second, each clean-safety obligation reduces to a commutativity check, yielding local constraints
$[U,Z_a]=0$ and $[U,X_a]=0$~\cite{QCQI}, where $X_a$ and $Z_a$ are the Pauli matrices on ancilla qubit $a$. 
This transforms a global state-space question into independent operator-level checks, enabling parallel verification and fine-grained diagnosis:
violations are classified as \texttt{LogicError} (leakage in the computational basis or in the $Z$ basis) or \texttt{PhaseError} (coherence/phase faults exposed in the $X$ basis).
Crucially, this algebraic formulation decouples the verification specification from the underlying symbolic engine. 
This backend-agnostic design enables the seamless integration of diverse solvers, such as \textsc{Quokka-Sharp}~\cite{quokka-sharp,Mei2024}, a quantum circuit analysis toolkit based on weighted model counting, and \textsc{QCEC}~\cite{qcec}, a state-of-the-art equivalence checker.
Finally, guided by the resulting diagnostics, we synthesize targeted local repairs.
While repair is inherently limited to separable error cases, it isolates and corrects common local deviations (e.g., residual phases or bit flips), restoring dirty safety using $O(1)$ single-qubit operations.

In summary, our main contributions are:
\begin{enumerate}
    \item \textbf{Unified reductions for dirty and clean ancilla safety.}
    We give an algebraic characterization that reduces ancilla safety of an $m$-qubit register to $2m$ independent commutativity checks, $[U,Z_a]=0$ and $[U,X_a]=0$.

    \item \textbf{Diagnostics and lightweight repair.}
    The checks yield interpretable failure modes (\texttt{LogicError} vs.\ \texttt{PhaseError}), which we use to drive a targeted repair pass that fixes separable violations via local single-qubit corrections.

    \item \textbf{Backend-agnostic implementation and large-scale evaluation.}
    We implement the pipeline with two complementary symbolic backends and evaluate it on structured and random benchmarks (including Grover~\cite{grover1996} and adders~\cite{gidney2018halving}), demonstrating scalability to thousands of qubits and validating end-to-end verify--repair--reverify. 
\end{enumerate}
To maintain the flow of the narrative, we defer the detailed information on examples, omitted proofs, and additional experimental details to the appendix.

\subsection{Related Work}
We briefly review the lines of research most closely connected to ancilla borrowing and dirty ancilla safety, and summarize how our approach fits into this landscape.

\paragraph{Dirty ancilla in compilation and synthesis.}
Reusing dirty ancilla qubits, ancilla qubits initialized in an arbitrary (possibly entangled) state, goes back to Barenco et al.~\cite{barenco1995elementary} and has become a standard technique for space-efficient compilation in the NISQ era~\cite{Preskill2018}.
It has been exploited to optimize arithmetic components underlying Shor-style routines, including adders and modular arithmetic~\cite{Haner2017,gidney2018halving,shor1997}, and more recently to reduce non-Clifford cost (e.g., T-depth) in unitary synthesis~\cite{Low2024}.
These optimizations come with a subtle semantic obligation: borrowed ancilla qubits must be returned to their original state, not merely reset to a fixed basis state, and this requirement can become difficult to track after aggressive rewriting or synthesis~\cite{Su2025}.
Our work treats dirty ancilla safety as an explicit semantic property of a compiled block and provides an automated procedure that applies directly to optimized circuits.

\paragraph{Existing verification frameworks.}
Ensuring ancilla safety has been addressed through diverse programming languages and verification paradigms. 
Deductive approaches based on quantum Hoare logic~\cite{ying2012floyd,ying2024foundations,feng2020hoare} and linear typing disciplines~\cite{qwire,reqwire} enforce disciplined qubit usage through structural patterns like the ``compute--uncompute'' discipline. 
Similarly, languages such as Q\#~\cite{svore2018qsharp} and Silq~\cite{bichsel2020silq} provide syntactic constructs or automated uncomputation to manage temporary qubits. 
While effective when cleanup structures are explicit, these methods struggle with heavily optimized or synthesis-generated blocks where safety properties are no longer syntactically manifest. 
Although Su et al.~\cite{Su2025} recently formalized dirty safety for restricted classical settings, our framework provides a gate-agnostic treatment for universal quantum circuits. 
By reducing safety checking to efficiently checkable algebraic conditions, we enable the certification of arbitrary circuit blocks without assuming a visible compute--uncompute pattern, while additionally supporting diagnostic witness generation.

\paragraph{From global uncomputation to lightweight correction.}
A standard way to ensure ancilla safety is \emph{constructive uncomputation}: append an inverse computation to disentangle and reset auxiliaries.
Repair frameworks such as Reqomp~\cite{Paradis2022Reqomp} and Unqomp~\cite{ParadisBSV2021uncompute} and language-level automation in Silq~\cite{bichsel2020silq} follow this principle and provide strong guarantees, but can incur substantial overhead (often close to doubling depth) and, by design, target \emph{clean} reset semantics.
In contrast, dirty safety requires preserving an arbitrary unknown initial ancilla state.
Guided by the diagnosis, our approach enables a \emph{local} correction strategy: for a useful class of separable faults, safety can be restored via targeted fixes (e.g., correcting residual phase behavior) rather than applying a global uncompute, providing a practical alternative when preserving the original ancilla context is essential.

\section{Preliminaries}
This section provides the notation and basic definitions from quantum computing used throughout this paper. For more details, please refer to~\cite{QCQI}.

\subsection{Quantum State and Circuit}\label{sec:prelim-qstate-circuit}
Computational tasks can be implemented by \emph{quantum circuits}, which consist of a sequence of quantum gates (computational steps) applied to an input quantum state that encodes information. We briefly introduce the relevant notions below.

\paragraph{Quantum state.}

The fundamental unit of quantum information is the \textit{qubit (quantum bit)}, represented as a unit vector in a two-dimensional complex Hilbert (linear) space $\mathcal{H}_2 \cong \mathbb{C}^2$.
We fix the computational basis $\{\ket{0},\ket{1}\}$, where $\ket{0}=(1,0)^\top$ and $\ket{1}=(0,1)^\top$, and $^\top$ denotes vector transpose.
A one-qubit \textit{pure state} $\ket{\psi}\in\mathcal{H}_2$ is a superposition $\ket{\psi}=\alpha\ket{0}+\beta\ket{1}$, where $\alpha,\beta\in\mathbb{C}$ and $|\alpha|^2+|\beta|^2=1$.
More generally, an $n$-qubit pure state lies in the $2^n$-dimensional Hilbert space $\mathcal{H}\cong\mathbb{C}^{2^n}$ and can be expressed as a superposition over the computational basis $
\big\{\ket{i_1}\otimes\ket{i_2}\otimes\cdots\otimes\ket{i_n}\ \big|\ i_j\in\{0,1\}\big\}$
often abbreviated as $\ket{i_1 i_2\cdots i_n}$ or simply $\ket{x}$ with $x\in\{0,1\}^n$, where $\otimes$ denotes the tensor product. To describe the state of a subsystem within such a composite space, we utilize the \textit{reduced density matrix} obtained via the \textit{partial trace} operation $\text{Tr}_B(\cdot)$, which is uniquely defined as the linear map satisfying $\text{Tr}_B(O_A \otimes O_B) = O_A \text{Tr}(O_B)$ for any operators $O_A$ on subsystem $A$ and $O_B$ on subsystem $B$.

\paragraph{Quantum circuits.}
A quantum circuit on an $n$-qubit \emph{register} (i.e., an ordered collection of $n$ qubits, with state space $\mathcal{H}_2^{\otimes n}$) is a finite, time-ordered sequence of quantum gates.
Each quantum gate is a unitary operator, i.e., a complex matrix $G$ satisfying $G^\dagger G = I$, where $G^\dagger$ is the conjugate transpose of $G$ and $I$ is the identity.
A gate may act on only a subset of qubits; formally, its action on the full register is obtained by tensoring with identities on the unaffected qubits (up to the chosen qubit ordering), e.g., $I\otimes\cdots\otimes G\otimes\cdots\otimes I$.
If the circuit applies gates $G_1,G_2,\ldots,G_L$ in time order, then the overall transformation is the unitary
$
U \;=\; G_L \cdots G_2 G_1.
$
Typical 1-qubit quantum gates are Pauli gates (2-by-2 matrices):
\[
X=\begin{bmatrix}0&1\\[2pt]1&0\end{bmatrix},\quad
Y=\begin{bmatrix}0&-i\\[2pt]i&0\end{bmatrix},\quad
Z=\begin{bmatrix}1&0\\[2pt]0&-1\end{bmatrix},\]
where $i$ is the imaginary unit.

Applying the circuit to an input state $\ket{\psi}$ produces the output state $\ket{\psi'} = U\ket{\psi}$.
We refer to $L$ as the \emph{gate count} (the number of gates in the sequence). 
Separately, the \emph{depth} of a circuit is the minimum number of sequential layers after grouping together gates that act on disjoint sets of qubits and can therefore be executed in parallel.

\begin{remark}
In this work, we focus on purely unitary quantum circuits without mid-circuit measurements or classical feed-forward operations. This ensures that the circuit evolution can be represented as a single unitary matrix $U$, which is the foundation for our commutativity-based verification framework.
\end{remark}

\subsection{Ancilla Qubit Safety}\label{sec:ancilla-safety}
In many compiled or hand-designed implementations, the circuit acts on $(n+m)$ qubits rather than only the $n$ ``data'' qubits: the additional $m$ qubits are \emph{ancilla qubits}.
Ancillae serve as temporary workspace introduced to facilitate decompositions of multi-qubit operations, satisfy connectivity constraints, or enable more efficient scheduling.
In practice, they can reduce the circuit depth (and/or the gate count) under hardware constraints.
Accordingly, we model the circuit as a unitary transformation
\[
U:\ \mathcal{H}_\mathcal{W}\otimes\mathcal{H}_\mathcal{A} \;\to\; \mathcal{H}_\mathcal{W}\otimes\mathcal{H}_\mathcal{A},
\]
where $\mathcal{H}_\mathcal{W}$ is the state space of the $n$-qubit working register and $\mathcal{H}_\mathcal{A}$ is the state space of the $m$-qubit ancilla register. The partition into working qubits $\mathcal{W}$ and ancilla qubits $\mathcal{A}$ is conceptual and depends on the intended computation (and the compiler strategy), rather than reflecting a hardware distinction.

However, using ancilla qubits comes with an extra semantic requirement: ancilla qubits are not part of the intended output, so they must not retain any unintended information after the circuit finishes.
Concretely, the circuit should (i) restore the ancilla register to its initial state (or, in the dirty case, to its original unknown state) and (ii) remove any residual entanglement between the working register and the ancilla register, so that the ancilla qubits can be safely reused in subsequent computations.
We call this requirement \emph{ancilla qubit safety}.
In this paper we consider two standard levels of ancilla safety: \textit{clean ancilla safety}, where ancilla qubits are initialized to a known state (typically $\ket{0}^{\otimes m}$), and \textit{dirty ancilla safety}, where ancilla qubits may start in an arbitrary unknown state (possibly entangled with the working qubits or qubits outside of the circuit) and must be returned \emph{exactly} to that state.

\paragraph{Clean ancilla safety.}
Clean ancilla qubits are initialized to the standard basis state $\ket{0}^{\otimes m}$.
A circuit is \emph{clean safe} if, regardless of the input on the working register, the ancilla register is restored to $\ket{0}^{\otimes m}$ at the end.

\begin{definition}\label{def_clean_safety}
Let $U$ be an $(n+m)$-qubit unitary acting on $n$ working qubits $\W$ and $m$ ancilla qubits $\A$.
We say that $U$ is \textbf{clean safe} on $\mathcal{A}$ if and only if
\begin{equation}\label{eq:def_clean}
\forall\,\ket{\psi}_\mathcal{W}\in\mathcal H_\mathcal{W},\ \exists\,\ket{\psi'}_\mathcal{W}\in\mathcal H_\mathcal{W}\ \mathrm{s.t.}\ 
U\big(\ket{\psi}_\mathcal{W}\otimes\ket{0}_\mathcal{A}\big)=\ket{\psi'}_\mathcal{W}\otimes\ket{0}_\mathcal{A} .
\end{equation}
\end{definition}

\cref{fig:clean_circuit} shows a typical ``mediated'' construction: the ancilla qubit $q_a$ transfers information between $q_0$ and $q_1$, and is then uncomputed so that it returns to $\ket{0}$ and decouples from the working qubits, satisfying~\cref{eq:def_clean}.
\paragraph{Dirty ancilla safety.}
Dirty ancilla qubits arise when we \emph{borrow} qubits that are not guaranteed to be in $\ket{0}$ (and may even be entangled with other parts of the computation).
Since these qubits may carry information needed later, we cannot assume any fixed initialization; instead, correctness requires that the circuit acts as the identity on the ancilla register and applies some unitary on the working register, uniformly for \emph{all} ancilla states.

\begin{figure}[t]
    \centering
    \begin{subfigure}[t]{0.45\textwidth}
        \centering
        \begin{quantikz}[row sep=3mm]
        \lstick{$q_0$} &\ctrl{2} & \qw      & \ctrl{2} & \qw \\
        \lstick{$q_1$} &\qw      & \targ{}  & \qw      & \qw \\
        \lstick{$q_a$} &\targ{}  & \ctrl{-1}& \targ{}  & \qw
        \end{quantikz}
        \caption{}
        \label{fig:clean_circuit}
    \end{subfigure}\hfil
    \begin{subfigure}[t]{0.45\textwidth}
        \centering
        \begin{quantikz}[row sep=3mm]
        \lstick{$q_0$}    & \ctrl{2} & \qw      & \ctrl{2} & \qw      & \qw \\
        \lstick{$q_1$}    & \qw      & \targ{}  & \qw      & \targ{}  & \qw \\
        \lstick{$q_a$}    & \targ{}  & \ctrl{-1}& \targ{}  & \ctrl{-1}& \qw
        \end{quantikz}
        \caption{}
        \label{fig:dirty_circuit}
    \end{subfigure}
    \caption{Bridge-GHZ circuits~\cite{Itoko2020} with a shared ancilla $q_a$: (left) clean-ancilla version; (right) dirty-ancilla version. See detail in Appendix~\ref{subsec:ghz_bridge}.}
    \label{fig:bridge-ghz-clean-dirty}
\end{figure}
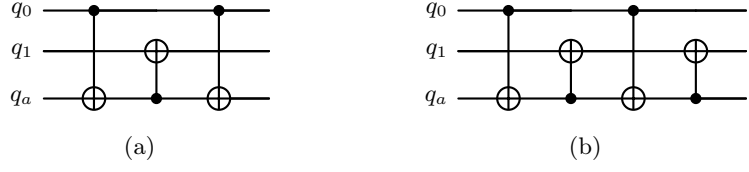
\begin{definition}\label{def_dirty_safety}
Let $U$ be an $(n+m)$-qubit unitary acting on $n$ working qubits $W$ and $m$ ancilla qubits $\A$.
We say that $U$ is \textbf{dirty safe} on $\A$ if and only if there exists a unitary $V$ on $\mathcal H_\mathcal{W}$ such that
\begin{equation}\label{eq:def_dirty}
\forall\,\ket{\psi}_\mathcal{W}\in\mathcal H_\mathcal{W},\ \forall\,\ket{\varphi}_\mathcal{A}\in\mathcal H_\mathcal{A}:\quad
U\big(\ket{\psi}_\mathcal{W}\otimes\ket{\varphi}_\mathcal{A}\big)=(V\ket{\psi}_\mathcal{W})\otimes\ket{\varphi}_\mathcal{A} .
\end{equation}
\end{definition}

Intuitively, \cref{eq:def_dirty} means that the effect of $U$ on the working register is independent of the (unknown) ancilla state, and the ancilla register is preserved exactly.
\cref{fig:dirty_circuit} illustrates this: the ancilla $q_\mathcal{A}$ may start in any state, but after the full sequence the working qubits implement the intended transformation while $q_\mathcal{A}$ is returned unchanged, hence the circuit is dirty safe.

\begin{remark}[Ancilla qubits at the boundary]\label{rem:boundary}
\cref{def_clean_safety,def_dirty_safety} assume (for notational convenience) that the ancilla register $\A$ forms a suffix of the qubit ordering.
This is without loss of generality: if $\A$ is interleaved with working qubits, we can apply a permutation (swap) unitary $\Pi_\mathcal{A}$ that moves the qubits in $\A$ to the end, and instead consider the conjugated circuit $U'=\Pi_\mathcal{A} U \Pi_\mathcal{A}^\dagger$.
Ancilla safety is invariant under such re-indexing, so we assume throughout that ancilla qubits are placed at the end.
\end{remark}

\section{Verification of Clean Ancilla Safety}\label{sec:clean}
In this section, we reduce the verification of \emph{clean ancilla safety} to checking a purely algebraic condition: the invariance of a designated ``clean'' subspace, which in turn is equivalent to a commutativity test between the circuit unitary and a fixed reflection operator.
This yields a practical verification procedure that can be discharged by standard matrix-equivalence (or commutativity) engines.
To maintain the flow of the narrative, we defer all proofs to the appendix~\ref{app:proof}.

\emph{Problem reduction outline.} We reduce clean ancilla safety checking to a single algebraic check.
The reduction proceeds in three steps: 
(i) we show that clean ancilla safety is equivalent to invariance of the clean-input subspace under the circuit unitary; 
(ii) we express subspace invariance via the corresponding orthogonal projector (equivalently, its \emph{Householder reflection}); and 
(iii) we obtain a commutativity obligation that can be discharged by standard equivalence-checking backends.

\paragraph{Setup.}
Let $U$ be an $(n+m)$-qubit unitary acting on $\mathcal{H}=\mathcal{H}_{\W}\otimes\mathcal{H}_{\A}$, where $\mathcal{H}_\W$ corresponds to $n$ working qubits and $\mathcal{H}_\A$ corresponds to $m$ ancilla qubits.
Define the \emph{clean subspace}
\[
\mathcal{V} \ :=\ \mathcal{H}_\W \otimes \mathrm{span}\{\ket{0}_\A\},
\]
where $\operatorname{span}\{\cdot\}$ denotes the linear span. 
Clean safety requires that every valid input $\ket{\psi}_\W\ket{0}_\A$ is mapped to an output whose ancilla register is still $\ket{0}_\A$ and is disentangled from the working register.
Geometrically, this means exactly that $U$ maps the clean subspace $\mathcal{V}$ onto itself.

\begin{lemma}[Clean safety $\Leftrightarrow$ clean-subspace invariance]\label{lem:clean-subspace}
Let $U$ and $\mathcal{V}$ be as above.  $U$ is clean safe on $\A$ if and only if $U$ preserves the clean subspace: $U(\mathcal{V})=\mathcal{V}$, where $
U(\mathcal{V}) \ :=\ \{\,U\ket{v} \mid \ket{v}\in \mathcal{V}\,\}
$
denotes the image of $\mathcal{V}$ under $U$.
   
\end{lemma}

\noindent
Having reduced clean ancilla safety to subspace invariance, we next translate $U(\mathcal{V})=\mathcal{V}$ into an operator identity that can be checked automatically using projectors and reflections.

\paragraph{Projector and reflection.}
Let $P_{\mathcal{V}}$ be the orthogonal projector onto $\mathcal{V}$.
Since $\mathcal{V}=\mathcal{H}_\W\otimes \mathrm{span}\{\ket{0}_\A\}$, we have
\[
P_{\mathcal{V}}= I_\W\otimes \ket{0}\!\bra{0}_\A,
\qquad
P_{\mathcal{V}}^2=P_{\mathcal{V}},
\qquad
P_{\mathcal{V}}^\dagger=P_{\mathcal{V}},
\]
where $I_\W$ is the identity operator onto the working qubit space $\mathcal{H}_\W$. 
We further define the associated Householder reflection~\cite{horn1985matrix}:
\[
Q_{\mathcal{V}}\ :=\ 2P_{\mathcal{V}}-I_{\W\cup \A},
\]
which is unitary and Hermitian ($Q_{\mathcal{V}}=Q_{\mathcal{V}}^\dagger$).
In particular, one can view $Q_{\mathcal{V}}$ as a fixed diagonal reflection that flips the phase of states outside $\mathcal{V}$ (and acts as $+1$ on $\mathcal{V}$).

\begin{lemma}[Invariance $\Leftrightarrow$ commutativity]\label{lemma:clean-commute} Let $U$, $\mathcal{V}$ and  $Q_{\mathcal{V}}$ be as above.
$U(\mathcal{V})=\mathcal{V}$ if and only if

$[U,Q_{\mathcal{V}}]=0,$ 
where $[U,Q_{\mathcal{V}}]\coloneqq UQ_{\mathcal{V}}-Q_{\mathcal{V}}U,$ the commutator of the two matrices.
\end{lemma}

 Combining this result with \cref{lem:clean-subspace}, we reduce clean ancilla safety checking to a purely algebraic check (commutativity check).
\begin{theorem}[Commutativity reduction for clean ancilla safety]\label{thm:clean-reduction}
$U$ is clean safe on $\A$ if and only if $[U,Q_{\mathcal{V}}]=0$.
\end{theorem}

\paragraph{Algorithmic consequence.}
By \cref{thm:clean-reduction}, verifying clean ancilla safety reduces to a commutativity check between the circuit unitary $U$ and the fixed reflection $Q_{\mathcal{V}}$ determined by the clean subspace.
Equivalently, one may check the projector invariance condition $UQ_{\mathcal{V}}U^\dagger=Q_{\mathcal{V}}$ by applying $U^\dagger$ on the left side of equation $[U,Q_{\mathcal{V}}]=0$; both are standard matrix-equivalence obligations.

We end this section by checking the clean ancilla safety of the circuit in \cref{fig:clean_circuit} using \cref{thm:clean-reduction}. 
\begin{example}\label{exa:clean-example}
We illustrate the reduction on the 3-gate circuit in \cref{fig:clean_circuit}, which implements a CNOT between working qubits $q_0,q_1$ mediated by an ancilla $q_a$.
The clean subspace is $\mathcal{V}=\mathcal{H}_{\{q_0,q_1\}}\otimes \mathrm{span}\{\ket{0}_{q_a}\}$, with projector  and reflection:
\[
P_{\mathcal{V}} = I_{q_0}\otimes I_{q_1}\otimes \ket{0}\!\bra{0}_{q_a}, \qquad Q_{\mathcal{V}} = 2P_{\mathcal{V}} - I = I_{q_0}\otimes I_{q_1}\otimes Z_{q_a}.
\]

\cref{thm:clean-reduction} states that the circuit is clean safe if and only if $[U,Q_{\mathcal{V}}]=0$ (equivalently, $UQ_{\mathcal{V}}U^\dagger=Q_{\mathcal{V}}$). We can conclude the clean ancilla safety of the circuit by verifying that the condition  $[U,Q_{\mathcal{V}}]=0$ holds. For the full calculation, please see the appendix~\ref{app:detail-example}.
\end{example}

\paragraph{Extension to arbitrary target ancilla states.}
For presentation, the clean ancilla safety development above uses the standard initialization $\ket{0}_\A$ and the corresponding clean subspace
$\mathcal{V}_0=\mathcal{H}_\W\otimes\mathrm{span}\{\ket{0}_\A\}$.
The same reduction applies verbatim to \emph{any} fixed target ancilla state $\ket{\varphi}_\A$.
Indeed, letting
\[
\mathcal{V}_{\varphi}:=\mathcal{H}_\W\otimes\mathrm{span}\{\ket{\varphi}_\A\},
\qquad
P_{\mathcal{V}_{\varphi}}:= I_\W\otimes \ket{\varphi}\!\bra{\varphi}_\A,
\qquad
Q_{\mathcal{V}_{\varphi}}:=2P_{\mathcal{V}_{\varphi}}-I_{\W\cup \A},
\]
we obtain the analogous commutativity characterization:
\begin{theorem}
\label{thm:clean-reduction-gen}
$U \ \text{is clean safe for }\ket{\varphi}_\A$ if and only if 
\[
[U,Q_{\mathcal{V}_{\varphi}}]=0, 
\quad Q_{\mathcal{V}_{\varphi}} = 2(I_\W \otimes \ket{\varphi}\!\bra{\varphi}_\A) - I.\]
\end{theorem}
Therefore, our clean ancilla safety checking procedure for $\ket{0}_\A$ is directly generalized to arbitrary target states $\ket{\varphi}_\A$ by replacing the projector and reflector $(P_{\mathcal{V}_0},Q_{\mathcal{V}_0})$ with $(P_{\mathcal{V}_{\varphi}},Q_{\mathcal{V}_{\varphi}})$.
This generality will be used in the dirty ancilla safety reduction in \cref{sec:dirty}, where we verify clean ancilla safety for multiple (linearly independent and non-orthogonal) ancilla states.

\section{Verification of Dirty Ancilla Safety}\label{sec:dirty}
In this section, we reduce \emph{dirty ancilla safety} check to a small number of \emph{clean ancilla safety} checks, and thus ultimately to the commutativity or equivalence obligations developed in \cref{sec:clean}.
The reduction has two layers: we first show that dirty ancilla safety on an $m$-qubit ancilla register is equivalent to dirty ancilla safety on each individual ancilla qubit; we then show that dirty ancilla safety for a \emph{single} ancilla qubit can be certified by checking state-dependent clean ancilla safety for two fixed, linearly independent and non-orthogonal ancilla states, equivalently two states satisfying $0<\left| \braket{\varphi_1}{\varphi_2} \right|< 1$ (e.g., $\ket{0}$ and $\ket{+}$).

\subsection{From Multi-qubit to Single-qubit Dirty Safety}\label{sec:dirty-multi-to-single}
We start with an important distinction between \emph{clean} and \emph{dirty} ancilla safety.
For clean ancilla qubits, safety is defined relative to a fixed initialization (typically $\ket{0}^{\otimes m}$), and the property is \emph{not} closed under marginalizing to individual qubits: a circuit may be clean safe for the whole register $\ket{0}^{\otimes m}$ while failing to be clean safe for a chosen ancilla qubit when the remaining qubits are treated as part of the working register.

In contrast, dirty ancilla safety is a stronger, basis-independent statement: the circuit must act as the identity on the ancilla register for \emph{all} ancilla states (including entangled ones).
As a result, dirty ancilla safety \emph{does} decompose cleanly across individual ancilla qubits.

\begin{lemma}[Multi-qubit dirty ancilla safety $\Leftrightarrow$ single-qubit dirty ancilla safety]\label{lem:dirty-multi-single}
Let $U$ be a unitary acting on working qubits $\W$ and ancilla qubits $\A=\{a_1,\dots,a_m\}$.
Then $U$ is dirty safe on the ancilla register $\A$ (\cref{def_dirty_safety}) if and only if, for every $a_i\in \A$, $U$ is dirty safe on $\{a_i\}$ when all other qubits are treated as working qubits, i.e.,
\[
\forall a_i\in \A:\ \exists\,V_i\ \text{on}\ \mathcal{H}_{\W\cup(\A\setminus\{a_i\})}\ \mathrm{s.t.}\ 
U = V_i \otimes I_{a_i}.
\]
\end{lemma}

\cref{lem:dirty-multi-single} implies that verifying dirty ancilla safety on an $m$-qubit ancilla register reduces to $m$ independent single-qubit dirty ancilla safety checks.

\subsection{From Single-qubit Dirty Ancilla Safety to Clean Ancilla Safety}\label{sec:dirty-single-to-clean}
We now focus on the case of a \emph{single} ancilla qubit $\A$.
Write $\mathcal{H}=\mathcal{H}_\W\otimes\mathcal{H}_\A$ with $\dim(\mathcal{H}_\A)=2$.
For any single-qubit state $\ket{\varphi}\in\mathcal{H}_\A$, define the corresponding ``clean'' subspace
\[
\mathcal{V}_{\varphi} := \mathcal{H}_\W\otimes \mathrm{span}\{\ket{\varphi}\},
\]
and recall that ``$U$ is clean safe for $\ket{\varphi}$'' means $U(\mathcal{V}_\varphi)=\mathcal{V}_\varphi$ (equivalently, $\forall\ket{\psi}_\W,\exists\ket{\psi'}_\W:\ U(\ket{\psi}_\W\ket{\varphi})=\ket{\psi'}_\W\ket{\varphi}$).

\paragraph{Orthogonal complement closure.}
For a single-qubit ancilla, let 
$\mathrm{span}\{\ket{\varphi}\}^{\perp}\subseteq \mathcal{H}_\A$
denote the one-dimensional subspace orthogonal to 
$\mathrm{span}\{\ket{\varphi}\}$.
Choose any unit vector $\ket{\varphi^\perp}$ spanning this subspace, i.e.,
$\braket{\varphi}{\varphi^\perp}=0$ and 
$\|\ket{\varphi^\perp}\|=1$.
Then the orthogonal complement of 
$\mathcal{V}_\varphi=\mathcal{H}_\W\otimes \mathrm{span}\{\ket{\varphi}\}$
is
\[
\mathcal{V}_\varphi^{\perp}
=
\mathcal{H}_\W\otimes \mathrm{span}\{\ket{\varphi^\perp}\}.
\]
For a single-qubit ancilla, invariance of $\mathcal{V}_\varphi$
therefore automatically implies invariance of $\mathcal{V}_{\varphi^\perp}$,
and vice versa.

\begin{lemma}[Closure under orthogonal complement]\label{lem:clean-orth}
Let $U$ act on $\mathcal{H}_\W\otimes\mathcal{H}_\A$ with a single ancilla qubit $\A$.
For any unit vector $\ket{\varphi}\in\mathcal{H}_\A$, choose any unit vector
$\ket{\varphi^\perp}\in\mathcal{H}_\A$ spanning
$\mathrm{span}\{\ket{\varphi}\}^{\perp}$.
Then
\[
U(\mathcal{V}_{\varphi})=\mathcal{V}_{\varphi}
\quad\text{if and only if}\quad
U(\mathcal{V}_{\varphi^\perp})=\mathcal{V}_{\varphi^\perp}.
\]
Equivalently, $U$ is clean safe for $\ket{\varphi}$ if and only if it is clean safe for $\ket{\varphi^\perp}$.
\end{lemma}

\paragraph{Two linearly independent and non-orthogonal states suffice.}
Dirty safety on a single ancilla qubit is the statement that $U$ acts trivially on the ancilla for \emph{all} ancilla states, i.e., $U=V\otimes I_\A$ for some unitary $V$ on the working register.
The following theorem shows that it suffices to test clean ancilla safety on two linearly independent and non-orthogonal ancilla states (which uniquely determine the action on the Bloch sphere).

\begin{lemma}[Single-qubit dirty ancilla safety via two clean checks]\label{lem:dirty-two-states}
Let $U$ be a unitary on $\mathcal{H}_\W\otimes\mathcal{H}_\A$ with a single ancilla qubit $\A$.
Fix any two linearly independent and non-orthogonal ancilla states $\ket{\varphi_1},\ket{\varphi_2}\in\mathcal{H}_\A$ with $0<\left| \braket{\varphi_1}{\varphi_2} \right|< 1$.
Then $U$ is dirty safe on $\A$ if and only if $U$ is clean safe for both $\ket{\varphi_1}$ and $\ket{\varphi_2}$, i.e.,
\[
U(\mathcal{V}_{\varphi_j})=\mathcal{V}_{\varphi_j},
\quad j\in \{1, 2\}.
\]
In particular, one may choose $\{\ket{\varphi_1},\ket{\varphi_2}\}=\{\ket{0},\ket+:=\tfrac{1}{\sqrt{2}}(\ket{0}+\ket{1})\}$.
\end{lemma}

\begin{theorem}[Commutativity reduction for dirty ancilla safety]\label{thm:dirty-reduction}
Let $U$ be a unitary acting on $\mathcal{H}=\mathcal{H}_\W\otimes\mathcal{H}_\A$, where
$\W$ is the working register and $\A=\{a_1,\dots,a_m\}$ is an $m$-qubit ancilla register.
Fix two linearly independent and non-orthogonal single-qubit states $\ket{\varphi_1},\ket{\varphi_2}\in\mathbb{C}^2$
with $0<\left| \braket{\varphi_1}{\varphi_2} \right|< 1$ (e.g., $\ket{\varphi_1}=\ket{0}$ and $\ket{\varphi_2}=\tfrac{1}{\sqrt{2}}(\ket{0}+\ket{1})$).
For each $a_i\in \A$ and $j\in\{1,2\}$, define
\[
P_{(a_i,\ket\varphi_j)}
\ :=\
I_{\W\cup(\A\setminus\{a_i\})}\otimes \ket{\varphi_j}\!\bra{\varphi_j}_{a_i},
\qquad
Q_{(a_i,\ket\varphi_j)}
\ :=\
2P_{(a_i,\ket\varphi_j)}-I.
\]
Then $U$ is dirty safe on $\A$ if and only if
\[
\forall\,a_i\in \A, j \in \{1,2\}:\quad
[U,\,Q_{(a_i,\ket\varphi_j)}]=0.
\]
Equivalently, $U$ is dirty safe on $\A$ if and only if
\[
\forall\,a_i\in \A, j \in \{1,2\}:\quad
U\,Q_{(a_i,\ket\varphi_j)}\,U^\dagger = Q_{(a_i,\ket\varphi_j)}.
\]
\end{theorem}

To illustrate \cref{thm:dirty-reduction}, we use it to show the dirty ancilla safety of circuits in~\cref{fig:clean_circuit} and \cref{fig:dirty_circuit}. 
Throughout, we choose the canonical test pair $\ket{0}$ and $\ket{+}=\tfrac{1}{\sqrt{2}}(\ket{0}+\ket{1})$, which are linearly independent and non-orthogonal.

\begin{example}[A circuit that is \emph{not} dirty safe (via commutativity checks)]\label{ex:dirty-unsafe}
Consider the circuit in \cref{fig:clean_circuit} with working register $\W=\{q_0,q_1\}$ and a single ancilla $\A=\{a\}$.
We apply \cref{thm:dirty-reduction} with the linearly independent and non-orthogonal test states
$\{\ket{\varphi_1},\ket{\varphi_2}\}=\{\ket{0},\ket+\}$.
For this single-ancilla case,
\[
\begin{aligned}
Q_{(a,\ket0)}
&=2(I_\W\otimes \ket{0}\!\bra{0}_a)-I
=I_\W\otimes Z_a,\\
Q_{(a,\ket+)}
&=2(I_\W\otimes \ket{+}\!\bra{+}_a)-I
=I_\W\otimes X_a.
\end{aligned}
\]
Therefore, dirty ancilla safety would require \emph{both} commutators to vanish:
\[[U,I_\W\otimes Z_a]=0,\quad [U,I_\W\otimes X_a]=0.\]

For the circuit $U=\mathrm{CNOT}_{q_0,a}\,\mathrm{CNOT}_{a,q_1}\,\mathrm{CNOT}_{q_0,a}$, one checks that
\[
U\,(I_\W\otimes Z_a)\,U^\dagger \;=\; I_\W\otimes Z_a,
\]
so $[U,Q_{(a,\ket0)}]=0$ holds.
However,
\[
U (I_\W \otimes X_a) U^\dagger
=
I_{q_0}\otimes X_{q_1}\otimes X_a
\neq
I_\W\otimes X_a .
\]
Equivalently $[U,Q_{(a,\ket+)}]\neq 0$.
By \cref{thm:dirty-reduction}, failing the $\ket{+}$ commutativity test implies that the circuit \emph{is not} dirty safe.
\end{example}

\begin{example}[A dirty safe circuit (via commutativity checks)]\label{ex:dirty_safe}
Now consider the 4-CNOT circuit in \cref{fig:dirty_circuit}, again with working register $W=\{q_0,q_1\}$ and ancilla $\A=\{a\}$.
Using the same test states $\ket{0}$ and $\ket{+}$, we again have
$Q_{(a,\ket0)}=I_\W\otimes Z_a$ and $Q_{(a,\ket+)}=I_\W\otimes X_a$.

Let $U$ denote the unitary implemented by \cref{fig:dirty_circuit}.
For this circuit, the extra uncomputation step ensures that both reflections are invariant:
\[
U\,(I_\W\otimes Z_a)\,U^\dagger \;=\; I_\W\otimes Z_a,
\qquad
U\,(I_\W\otimes X_a)\,U^\dagger \;=\; I_\W\otimes X_a.
\]
Equivalently, $[U,Q_{(a,\ket0)}]=0$ and $[U,Q_{(a,\ket+)}]=0$.
Hence, by \cref{thm:dirty-reduction}, the circuit in \cref{fig:dirty_circuit} \emph{is} dirty safe.
\end{example}

\section{Verification and Repair Algorithms}
In this section, we present an end-to-end workflow for ancilla qubit safety.
Using the commutativity characterizations proved earlier, we derive a practical dirty ancilla safety verifier (\cref{alg:dirty_safe-verification}) that performs two commutativity checks per ancilla qubit, a Pauli-$Z$ check and a Pauli-$X$ check corresponding to the test states $\ket{0}$ and $\ket{+}$, through a backend-agnostic primitive \textsc{CheckCommutativity}.

A case study illustrates how this dual check diagnoses failures as either \texttt{LogicError} or \texttt{PhaseError}.
Finally, we exploit this diagnosis to propose a lightweight local-rotation repair routine (\cref{alg:deterministic-repair}) for faults that can be corrected by appending single-qubit gates on the offending ancilla.

\subsection{Verification}

Our verifier (\cref{alg:dirty_safe-verification}) is a direct instantiation of the commutativity reduction
in \cref{thm:dirty-reduction} together with the commutativity
characterization of clean ancilla safety in \cref{thm:clean-reduction}.
Throughout, we use the canonical test states
\(
\ket{\varphi_1}=\ket{0}
\)
and
\(
\ket{\varphi_2}=\ket{+}
\), which are linearly independent and non-orthogonal,
so that the associated Householder reflections coincide with Pauli operators (up to an irrelevant global phase): $
Q_{(a,\ket0)} =\; I_{\W\cup(\A\setminus\{a\})}\otimes Z_a$ and   $Q_{(a,\ket+)} =\; I_{\W\cup(\A\setminus\{a\})}\otimes X_a.$

Therefore, for each ancilla qubit $a\in A$, dirty ancilla safety is equivalent to the conjunction of two commutativity checks:
$
[U,\; I_{\W\cup(\A\setminus\{a\})}\otimes Z_a]=0$ and 
$[U,\; I_{\W\cup(\A\setminus\{a\})}\otimes X_a]=0.
$

\textit{Diagnosis.}
We report violations using two coarse-grained error types aligned with the two obligations above.
The \emph{$Z$-check} is sensitive to computational basis leakage of the target ancilla (failure to preserve the $\ket{0}/\ket{1}$ eigenspaces on $a$), and we report this as \texttt{LogicError}.
The \emph{$X$-check} is sensitive to coherence/phase faults on superposition inputs (failure to preserve the $\ket{+}/\ket{-}$ eigenspaces on $a$), including the case where the ancilla acquires a relative phase (e.g., $\ket{+}\mapsto\ket{-}$) or becomes correlated with the working register in the $X$ basis; we report this as \texttt{PhaseError}.
Passing both checks for every $a\in A$ certifies dirty ancilla safety.

\textit{Backend abstraction.}
We encapsulate the commutativity test (checking clean ancilla safety in \cref{thm:clean-reduction}) in a primitive
\textsc{CheckCommutativity}$(U,Q)$ that returns \textbf{True} if and only if $UQ=QU$
(equivalently $UQU^\dagger=Q$).
A backend may implement this contract by explicit matrix algebra~\cite{Jones_2019}, decision diagrams~\cite{Wille_2022,miller2006qmdd}, tensor networks~\cite{Markov_2008,Fishman_2022,Kissinger_2020}, or any other symbolic engine. Some of them will be shown in the next section for numerical experiments.
As we will see in the experimental results, the choice of backend can significantly affect runtime, even though the verification logic remains the same.

\begin{algorithm}[t]
\caption{\textsc{DirtySafetyVerifier}$(U,\W,\A)$}
\label{alg:dirty_safe-verification}
\KwIn{Unitary circuit $U$, working qubits $\W$, ancilla qubits $\A$.}
\KwOut{\textbf{True} indicates all ancilla qubits in $\A$ are dirty safe, otherwise \textbf{False} and an \textsf{ErrorList}.}

$res\gets \textbf{True}$;\quad $\textsf{ErrorList}\gets[\ ]$\;

\ForEach{$a\in \A$}{
  $Q_Z \gets I_{\W\cup(\A\setminus\{a\})}\otimes Z_a$\;
  \If{\textsc{CheckCommutativity}$(U,Q_Z)$ is \textbf{False}}{
    $res\gets \textbf{False}$\;
    Add $(a,\texttt{LogicError})$ to $\textsf{ErrorList}$\;
  }

  $Q_X \gets I_{\W\cup(\A\setminus\{a\})}\otimes X_a$\;
  \If{\textsc{CheckCommutativity}$(U,Q_X)$ is \textbf{False}}{
    $res\gets \textbf{False}$\;
    Add $(a,\texttt{PhaseError})$ to $\textsf{ErrorList}$\;
  }
}
\Return $res,\textsf{ErrorList}$
\end{algorithm}

\textit{Correctness and complexity.}
Let $N$ be the total qubit number and $m=|\A|$.
By \cref{thm:dirty-reduction}, the algorithm is sound and complete:
it returns \textbf{True} if and only if $U$ is dirty safe on $\A$.
It performs exactly $2m$ commutativity checks; thus, its runtime is
$O(m\cdot T_{\mathrm{comm}}(N))$ where $T_{\mathrm{comm}}(N)$ is the cost of one call to
\textsc{CheckCommutativity} under the chosen backend.

\medskip
We next demonstrate how \cref{alg:dirty_safe-verification} separates
computational basis leakage from coherence failures, illustrating why both checks are necessary.

\begin{figure}[ht]
    \centering
    \begin{quantikz}[row sep=3mm]
    \lstick{$q_0$} & \ctrl{3} & \qw      & \ctrl{3} & \qw      & \qw      & \qw      & \qw      & \qw      & \qw                & \qw \\
    \lstick{$q_1$} & \qw      & \targ{}  & \qw      & \targ{}  & \ctrl{3} & \qw      & \ctrl{3} & \qw      & \qw                & \qw \\
    \lstick{$q_2$} & \qw      & \qw      & \qw      & \qw      & \qw      & \targ{}  & \qw      & \targ{}  & \qw                & \qw \\
    \lstick{$a_0$} & \targ{}  & \ctrl{-2}& \targ{}  & \ctrl{-2}& \qw      & \qw      & \qw      & \qw      & \gate[style={draw,densely dotted,rounded corners}]{Z} & \qw \\
    \lstick{$a_1$} & \qw      & \qw      & \qw      & \qw      & \targ{}  & \ctrl{-2}& \targ{}  & \ctrl{-2}& \qw                & \qw 
    \end{quantikz}
    \caption{A polluted cascaded-mediation circuit: ancilla $a_1$ is unchanged, while ancilla $a_0$ is corrupted by an injected $Z$ gate at the end of its cycle.}
    \label{fig:Double CX Z dirty}
\end{figure}
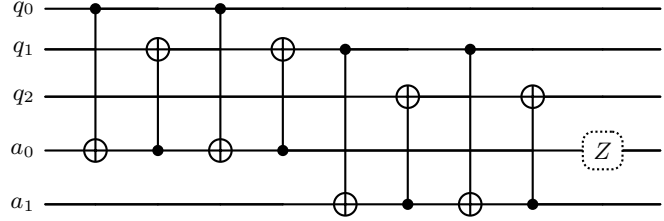

\begin{example}[Diagnosing a hidden \texttt{PhaseError}]\label{ex:phase-fault} Consider the cascaded mediation circuit of \cref{fig:Double CX Z dirty},
where we inject a single $Z$ gate on ancilla $a_0$ at the end of its cycle.
Let $U_{\mathrm{ideal}}$ denote the error-free unitary and
\(
U_{\mathrm{polluted}} := Z_{a_0}\,U_{\mathrm{ideal}}
\)
denote the polluted one.

\begin{itemize}
    \item \emph{Ancilla $a_1$ (baseline)}. The sub-circuit for $a_1$ has the standard compute--uncompute form, hence both commutators vanish: 
    \[
    [U_{\mathrm{ideal}},\ I\otimes Z_{a_1}]=0
        \quad\text{and}\quad
        [U_{\mathrm{ideal}},\ I\otimes X_{a_1}]=0.
    \]

    \cref{alg:dirty_safe-verification} therefore reports no error on $a_1$.
    \item \emph{Ancilla $a_0$ (injected $Z$ error)}. For the polluted circuit, the $Z$-check still passes:
    \[
    [U_{\mathrm{polluted}},\ I\otimes Z_{a_0}]
        =
        [Z_{a_0}U_{\mathrm{ideal}},\ Z_{a_0}]
        =
        Z_{a_0}[U_{\mathrm{ideal}},Z_{a_0}] + [Z_{a_0},Z_{a_0}]U_{\mathrm{ideal}}
        =
        0.
    \]

    so the verifier raises \emph{no} \texttt{LogicError} on $a_0$ (the injected $Z$ does not create computational basis leakage).

    However, the $X$-check fails because $Z_{a_0}$ anti-commutes with $X_{a_0}$:

    \[
    \begin{aligned}
    [U_{\mathrm{polluted}},\ I\otimes X_{a_0}]
    &= Z_{a_0}[U_{\mathrm{ideal}},X_{a_0}]
       + [Z_{a_0},X_{a_0}]U_{\mathrm{ideal}} \\
    &= [Z_{a_0},X_{a_0}]U_{\mathrm{ideal}} \\
    &\neq 0.
    \end{aligned}
    \]

    Hence \cref{alg:dirty_safe-verification} flags $a_0$ with a \texttt{PhaseError},
    capturing the fact that a trailing $Z$ gate is invisible on $\ket{0}$-inputs but flips the relative phase on $\ket{+}$-inputs (e.g., $\ket{+}\mapsto\ket{-}$).
\end{itemize}
\end{example}

This example shows why a single-basis (only-$Z$) check is insufficient for dirty ancilla safety:
it can miss “silent” phase faults. The dual check mandated by \cref{thm:dirty-reduction}
is necessary to rule out both computational basis leakage (\texttt{LogicError}) and coherence/phase failures (\texttt{PhaseError}).

\subsection{Repairing Unsafe Ancilla Qubits by Local Rotations}\label{sec:repair}

According to the error types (\texttt{LogicError} and \texttt{PhaseError}), we propose a deterministic repair routine (\cref{alg:deterministic-repair}) for the subset of unsafe dirty ancilla qubits whose violations are \emph{local} to the ancilla and do not entangle it with the rest of the circuit. 
This restriction is deliberate: local repair is intended as a lightweight post-compilation correction when the residual fault is genuinely ancilla-local, rather than as a replacement for global uncomputation or resynthesis. 
In this separable case, the error manifests as a single-qubit unitary distortion on the ancilla, and can therefore be corrected by appending only single-qubit rotations, without touching the working register.

\paragraph{Separable pre-check.}
A necessary prerequisite for local repair is that the verifier’s witness remains confined to the target ancilla.
Concretely, for $\ket\varphi\in\{\ket0,\ket+\}$, let
\[
Q_{(a,\ket\varphi)} \;:=\; I_{\W\cup(\A\setminus\{a\})}\otimes(2\ket\varphi\!\bra\varphi-I_a),
\qquad
Q'_{(a,\ket\varphi)} \;:=\; U\,Q_{(a,\ket\varphi)}\,U^\dagger.
\]
The error on $a$ is \emph{separable} if $Q'_{(a,\ket\varphi)}$ factorizes as
\[
Q'_{(a,\ket\varphi)} \;=\; I_{\W\cup(\A\setminus\{a\})}\otimes U_{(a,\ket\varphi)}
\]
for some single-qubit unitary $U_{(a,\ket\varphi)}$, i.e., the evolved operator has no support outside $a$ (equivalently, the error does not create entanglement or correlation between $a$ and the other qubits)~\cite{QCQI}.
If this pre-check fails, no local unitary on $a$ can decouple the registers, so \cref{alg:deterministic-repair} immediately reports $a$ as irreparable and adds it to \textsf{FailList} (the separability-check branch in the algorithm).

\paragraph{Repair strategy.}
\Cref{alg:deterministic-repair} proceeds ancilla-by-ancilla.
It first diagnoses each $a$ by invoking \textsc{DirtySafetyVerifier} (Algorithm~\ref{alg:dirty_safe-verification}) on the singleton set $\{a\}$ (Algorithm~\ref{alg:deterministic-repair}, Line~\ref{alg:each_a}); if the check passes, we skip $a$.
Otherwise, after the separability pre-check, the repair is chosen deterministically from the reported syndrome:
(i) for an isolated \texttt{LogicError}, we append a single $R_X(\theta)$;
(ii) for an isolated \texttt{PhaseError}, we append a single $R_Z(\theta)$;
and (iii) when both errors are present, we synthesize a three-gate $Z$--$X$--$Z$ Euler patch
$R_Z(\theta_{z_1})R_X(\theta_x)R_Z(\theta_{z_2})$ that corrects an arbitrary single-qubit unitary in $\mathrm{SU}(2)$
(Algorithm~\ref{alg:deterministic-repair}, Line~\ref{ln:repair:hybrid}).
All angles are computed analytically from the locally evolved operators (details in Appendix~\ref{app:implementation}).

Finally, the algorithm reruns \textsc{DirtySafetyVerifier} (Algorithm~\ref{alg:dirty_safe-verification})  on $\{a\}$ to validate the patch and commits it only if the post-check succeeds (the last two lines inside the loop); otherwise, it discards the patch and adds $a$ to \textsf{FailList}.

\begin{algorithm}[ht]
\caption{\textsc{LocalRepair}$(U,\W,\A)$}
\label{alg:deterministic-repair}
\KwIn{Circuit $U$, working qubits $\W$, ancilla qubits $\A$.}
\KwOut{Repaired circuit $U_{\mathrm{fix}}$, irreparable ancilla qubits \textsf{FailList}.}

$U_{\mathrm{fix}}\gets U$;\quad $\textsf{FailList}\gets[\ ]$\;

\BlankLine
\ForEach{$a\in \A$}{ \label{alg:each_a}
  $(\textrm{IsSafe},\ \textrm{Err}) \gets \textsc{DirtySafetyVerifier}(U_{\mathrm{fix}},\W\cup (\A\setminus\{a\}),\{a\})$\;
  
  \If{$\textrm{IsSafe}$}{\textbf{continue}}

  \BlankLine
  \If{$Q'_{(a,\ket0)}$ or $Q'_{(a,\ket+)}$ is not separable}{
    Add $a$ to \textsf{FailList};\quad \textbf{continue}
  }

  \BlankLine
  $U_{\mathrm{tmp}} \gets U_{\mathrm{fix}}$\;
  
  \uIf{\texttt{LogicError} $\in$ \textrm{Err} and \texttt{PhaseError} $\in$ \textrm{Err}}{
      \label{ln:repair:hybrid}
      Compute $\theta_{z_1}, \theta_x, \theta_{z_2}$ analytically\;
      $U_{\mathrm{tmp}} \gets (I \otimes R_Z(\theta_{z_1})_a R_X(\theta_x)_a R_Z(\theta_{z_2})_a)\cdot U_{\mathrm{tmp}}$\;
  }
  \uElseIf{\texttt{LogicError} $\in$ \textrm{Err}}{
    Compute $\theta$ analytically\;
    $U_{\mathrm{tmp}} \gets (I\otimes R_X(\theta)_a)\cdot U_{\mathrm{tmp}}$\;
  }
  \uElseIf{\texttt{PhaseError} $ \in \textrm{Err}$}{
    Compute $\theta$ analytically\;
    $U_{\mathrm{tmp}} \gets (I\otimes R_Z(\theta)_a)\cdot U_{\mathrm{tmp}}$\;
  }

  \BlankLine
  $(\textrm{IsSafe}_{\mathrm{new}},\ \_) \gets \textsc{DirtySafetyVerifier}(U_{\mathrm{tmp}},\W\cup (\A\setminus\{a\}),\{a\})$\;
  
  \eIf{$\textrm{IsSafe}_{\mathrm{new}}$}{
      $U_{\mathrm{fix}} \gets U_{\mathrm{tmp}}$\;
  }{
      Add $a$ to \textsf{FailList};
  }
}
\Return $U_{\mathrm{fix}}, \textsf{FailList}$
\end{algorithm}

\paragraph{Complexity.}

Let $N$ be the total number of qubits, and $m=|\A|$. 
Each invocation of Algorithm~\ref{alg:dirty_safe-verification} is restricted to a single target ancilla qubit. Hence, the runtime is dominated by the backend costs for (i) running Algorithm~\ref{alg:dirty_safe-verification} on a single ancilla ($O( T_{\mathrm{comm}}(N))$) and (ii) extracting the local operator or state information needed for separability and angle synthesis ($T_{\mathrm{loc}}(N)$).
Angle computation itself is constant-time once the $2\times2$ local information is available.
In the worst case, each ancilla triggers a constant number of backend calls, so the total cost is
\[
O\!\big(m\cdot T_{\mathrm{comm}}(N)+ m\cdot T_{\mathrm{loc}}(N)\big),
\]
i.e., linear in the number of ancillae, avoiding iterative search over continuous parameters.

We close this subsection with a minimal end-to-end example showing that the deterministic repair can analytically eliminate the ``silent'' $Z$-type phase error from Example~\ref{ex:phase-fault}.

\begin{example}[Analytic repair of the hidden phase error]\label{ex:repair-phase-fault-deterministic}
Consider the polluted circuit from Example~\ref{ex:phase-fault},
$U_{\mathrm{polluted}} := Z_{a_0}U_{\mathrm{ideal}}$.
Algorithm~\ref{alg:dirty_safe-verification} flags $a_0$ with a \texttt{PhaseError}.
Since the injected $Z_{a_0}$ is local, the separability pre-check in \cref{alg:deterministic-repair} succeeds.
The \texttt{PhaseError} branch then computes $\theta=\pi$, yielding $R_Z(\pi)=-iZ$ as the correcting gate.
Appending it cancels the injected error:
\[
R_Z(\pi)_{a_0}\cdot U_{\mathrm{polluted}}
=
(-iZ_{a_0})\cdot Z_{a_0}U_{\mathrm{ideal}}
=
(-i)\,U_{\mathrm{ideal}}.
\]
Thus the circuit is restored to $U_{\mathrm{ideal}}$ up to a global phase, and the post-check in \cref{alg:deterministic-repair} accepts the repair.
\end{example}

\section{Experiment and Results}
\label{sec:exp}

In this section, we implement the full verify--repair--reverify pipeline (\cref{alg:dirty_safe-verification} and~\cref{alg:deterministic-repair}) in a prototype tool. All source files are available at the GitHub project~\url{https://github.com/Veri-Q/Ancilla-Safety}.

We validate the tool on diverse circuits ranging from arithmetic benchmarks to Grover's algorithm.
By invoking \cref{thm:dirty-reduction}, the verification of an $m$-qubit ancilla register  strictly decomposes into $2m$ independent verification tasks.
Consequently, the total verification overhead scales linearly with the number of target ancilla qubits $m$ (in terms of the number of backend invocations).
Given this decomposition, our experimental evaluation focuses on profiling the performance of verifying a single ancilla qubit, as the end-to-end runtime for an $m$-ancilla register is simply the aggregation of these independent checks.

\paragraph{Backends and oracle.}
Our verification conditions are \emph{backend-agnostic}: in the framework, ancilla safety is reduced to fixed semantic checks.
Any sound equivalence or satisfiability engine can be plugged in to discharge them.
In our prototype, we instantiate this interface with two complementary symbolic backends: \textsc{Quokka-Sharp}, a quantum circuit analysis toolkit based on weighted model counting, and \textsc{QCEC}~\cite{qcec}, a state-of-the-art verification tool for equivalence checking.
For validation, we also implement a direct matrix-simulation oracle.
Although limited to small circuits due to exponential memory, it provides ground-truth verdicts for differential testing and sanity checks.

\paragraph{Platform.}
All experiments were run on a laptop with a 13th Gen Intel Core i9-13900H CPU and 32 GB RAM.

\paragraph{Benchmarks.}
We use a benchmark suite spanning structured algorithms and randomized circuits:
\begin{enumerate}
    \item \textit{Grover's search:} 
    We select Grover's algorithm as a representative benchmark for complex quantum programs. 
    This choice is motivated by the fact that Grover's search encapsulates the typical behaviors of advanced quantum routines: most complex algorithms either mirror Grover's patterns in utilizing \textit{dirty ancillae} or do not rely on ancillae significantly. 
    Our evaluation covers full circuits, single-iteration blocks, and variable-$N$/variable-round configurations~\cite{grover1996}, all of which exhibit a highly regular iterative structure. 
    In the following sections, $\text{Grover}_{r1}$ specifically denotes the Grover's algorithm configured with a single Grover cycle.
    \item \textit{Arithmetic and control logic:} adders~\cite{Haner2017}, Bridge GHZ~\cite{Itoko2020} (Appendix~\ref{subsec:ghz_bridge}), and multi-controlled-$X$ (MCX), including a clean safe version generated by Reqomp~\cite{Paradis2022Reqomp} and a dirty safe MCX implementation~\cite{mcxhttps}.
    \item \textit{Random circuits:} (i) uniformly random universal circuits and (ii) Qiskit-generated identity circuits (up to global phase).
\end{enumerate}

\subsection{Soundness and Scalability of Verification}

We evaluate two aspects of the verification process~\cref{alg:dirty_safe-verification}: \emph{soundness} and \emph{scalability}.
Soundness is assessed by differential testing against the matrix oracle in small instances.
Scalability is assessed in large circuits where matrix simulation is infeasible.
Runtimes are reported in seconds. We set a timeout of 3600\,s and a memory limit of 32\,GB.

In \cref{tab:final_unified_comparison}, $N$ is the total number of qubits, including working qubits and ancilla qubits. Memout and Timeout indicate that the verification process was terminated due to memory exhaustion and time limit violations, respectively. True denotes that the circuit was verified as safe, while the output: \texttt{LogicError} and \texttt{PhaseError} indicates that it was verified as unsafe, and corresponding errors are detected, \texttt{BothError} indicates that both errors are detected. A dash (--) represents that no output was generated. The table shows that the matrix oracle quickly runs into Memout as the number of qubits $N$ grows, while both symbolic backends scale to substantially larger instances. Across all oracle-comparable cases, the verdicts match, supporting the correctness of the implementations. The complete experimental data are provided in Appendix~\ref{app:full_data_unified_final}.

\begin{table*}[htbp]
\centering
\caption{\textbf{Comparison of verification efficiency and results.} Time is in seconds. ``--'' indicates the step is skipped due to timeout/memout.}
\label{tab:final_unified_comparison}
\renewcommand{\arraystretch}{0.9} 
\setlength{\tabcolsep}{4pt}  

\resizebox{\textwidth}{!}{
\begin{tabular}{@{}l cc rr rr rr@{}}
\toprule
\multirow{2}{*}{\textbf{Benchmark}} & \multirow{2}{*}{\boldmath$N$} & \multirow{2}{*}{\textbf{Depth}} & \multicolumn{2}{c}{\textbf{Matrix}} & \multicolumn{2}{c}{\textbf{Quokka-Sharp}} & \multicolumn{2}{c}{\textbf{QCEC}} \\
\cmidrule(lr){4-5} \cmidrule(lr){6-7} \cmidrule(lr){8-9}
 & & & \textbf{Time} & \textbf{Output} & \textbf{Time} & \textbf{Output} & \textbf{Time} & \textbf{Output} \\
\midrule

\multirow{3}{*}{Grover}
 & 7  & 55      & 0.1 & \stSafe  & 6.4 & \stSafe      & 0.2 & \stSafe  \\
 & 9  & 99      & 1.0 & \stSafe  & Timeout & \dash    & 0.3 & \stSafe  \\
 & 39 & 115779   & Memout & \dash & Timeout & \dash    & 25.9 & \stSafe  \\
\midrule

\multirow{2}{*}{Grover$_{r_1}$}
 & 319 & 1267 & Memout & \dash & 68.3 & \stSafe  & 5.5 & \stSafe  \\
 & 559 & 2227 & Memout & \dash & 223.8 & \stSafe & 15.6 & \stSafe \\
\midrule

\multirow{3}{*}{Adder}
 & 13 & 44 & 160.0 & \stSafe & 0.1 & \stSafe & 0.6 & \stSafe \\
 & 1499 & 5988  & Memout & \dash & 20.2 & \stSafe      & 81.1 & \stSafe      \\
 & 5999 & 23988 & Memout & \dash & 317.8 & \stSafe & Memout & \dash \\
\midrule

\multirow{3}{*}{\shortstack[l]{Bridge\\GHZ}}
 & 1399 & 2797  & Memout & \dash & 2.5 & \stSafe & 27.7 & \stSafe \\
 & 1999 & 3997  & Memout & \dash & 5.8 & \stSafe & 67.1 & \stSafe \\
 & 4799 & 9597  & Memout & \dash & 11.9 & \stSafe & Memout & \dash \\
\midrule

\multirow{3}{*}{\shortstack[l]{Bridge GHZ\\(Phase Err.)}}
 & 11 & 21 & 3.7 & \stPhase & 0.1 & \stPhase & 0.3 & \stPhase \\
 & 13 & 25 & 109.3 & \stPhase & 0.1 & \stPhase & 0.4 & \stPhase \\
 & 15 & 29 & Memout & \dash   & 0.1 & \stPhase & 0.4 & \stPhase \\
\midrule

\multirow{3}{*}{\shortstack[l]{Bridge GHZ\\(Logic Err.)}}
 & 11 & 21 & 3.6 & \stLogic & 0.1 & \stLogic & 0.3 & \stLogic \\
 & 13 & 25 & 104.9 & \stLogic & 0.1 & \stLogic & 0.3 & \stLogic \\
 & 15 & 29 & Memout & \dash   & 0.1 & \stLogic & 0.3 & \stLogic \\
\midrule

\multirow{3}{*}{\shortstack[l]{Reqomp-\\MCX}}
 & 11 & 9 & 3.2 & \stPhase & 0.1 & \stPhase & 0.3 & \stPhase \\
 & 13 & 11 & 96.1 & \stPhase & 0.1 & \stPhase & 0.4 & \stPhase \\
 & 15 & 13 & Memout & \dash   & 0.1 & \stPhase & 0.3 & \stPhase \\
\midrule

\multirow{3}{*}{MCX}
 & 1599 & 3192  & Memout & \dash & 28.0 & \stSafe  & 50.2 & \stSafe \\
 & 1799 & 3592  & Memout & \dash & 33.4 & \stSafe  & 63.9 & \stSafe \\
 & 5999 & 11992 & Memout & \dash & 414.4 & \stSafe & Memout & \dash \\
\midrule

\multirow{3}{*}{\shortstack[l]{Identity\\Random\\Circuit}}
 & 5  & 47 & 0.1 & \stSafe & 0.4 & \stSafe & 0.2 & \stSafe \\
 & 6  & 53 & 0.1 & \stSafe & 1.4 & \stSafe & 0.2 & \stSafe \\
 & 7  & 54 & 0.2 & \stSafe & 5.3 & \stSafe & 0.2 & \stSafe \\
 \midrule

\multirow{3}{*}{\shortstack[l]{Random\\Circuit}}
 & 30 & 11 & Memout & \dash & 0.1 & \stHybrid & 0.3 & \stHybrid \\
 & 50 & 14 & Memout & \dash & 0.1 & \stHybrid & 0.4 & \stHybrid \\
 & 70 &  4 & Memout & \dash & 0.1 & \stHybrid & 0.5 & \stHybrid \\
\bottomrule
\end{tabular}
}
\end{table*}

\begin{figure}[htpb]
    \centering
    \begin{tikzpicture}
    \begin{groupplot}[
        group style={group size=2 by 1, horizontal sep=1.2cm},
        grid=major,
        width=6.2 cm,
        height=4.5 cm,
        every axis title/.append style={at={(0.5,1.4)}, anchor=south},
    ]

    \nextgroupplot[
        xlabel={Total Qubits},
        ylabel={Time (s)},
        title={MCX},
        legend style={at={(0.5,1.05)}, anchor=south, legend columns=1, font=\footnotesize},
        legend cell align={left},
    ]
    \addplot[color=blue, mark=o, thick] coordinates {
        (199, 0.43) (599, 3.17) (999, 9.03) (1399, 16.71) (1599, 22.46) (1799, 28.26) (1999, 33.41)
    };
    \addlegendentry{Quokka-Sharp}
    \addplot[color=red, mark=x, thick] coordinates {
        (199, 1.89) (599, 8.69) (999, 15.64) (1399, 29.74) (1599, 43.42)
    };
    \addlegendentry{QCEC}

    \nextgroupplot[
        xlabel={Rounds ($R$)},
        title={Grover},
        ymode=log,
        legend style={at={(0.5,1.05)}, anchor=south, legend columns=2, font=\tiny},
        legend cell align={left},
    ]
    \addplot[color=blue, mark=o, thick] coordinates {
        (1, 0.03) (2, 0.07) (3, 0.23) (4, 0.40) (5, 0.60) (6, 0.83) (7, 1.06) (8, 1.64)
    };
    \addlegendentry{N=5 (Quokka-Sharp)}
    \addplot[color=red, mark=x, thick] coordinates {
        (1, 0.17) (2, 0.16) (3, 0.17) (4, 0.16) (5, 0.16) (6, 0.17) (7, 0.18) (8, 0.17)
    };
    \addlegendentry{N=5 (QCEC)}

    \addplot[color=blue, mark=square, thick] coordinates {
        (1, 0.06) (2, 0.87) (3, 6.26) (4, 16.92) (5, 22.68) (6, 35.59) (7, 51.35) (8, 220.40)
    };
    \addlegendentry{N=7 (Quokka-Sharp)}
    \addplot[color=red, mark=Mercedes star, thick] coordinates {
        (1, 0.18) (2, 0.18) (3, 0.20) (4, 0.21) (5, 0.21) (6, 0.19) (7, 0.21) (8, 0.19)
    };
    \addlegendentry{N=7 (QCEC)}

    \addplot[color=blue, mark=pentagon, thick] coordinates {
        (1, 0.06) (2, 2.82) (3, 110.51)
    };
    \addlegendentry{N=9 (Quokka-Sharp)}
    \addplot[color=red, mark=10-pointed star, thick] coordinates {
        (1, 0.19) (2, 0.23) (3, 0.25) (4, 0.22) (5, 0.21) (6, 0.21) (7, 0.21) (8, 0.22)
    };
    \addlegendentry{N=9 (QCEC)}

    \end{groupplot}
    \end{tikzpicture}
    \caption{Quokka-Sharp and QCEC performance on MCX and Grover circuits.}
    \label{fig:2performance}
\end{figure}
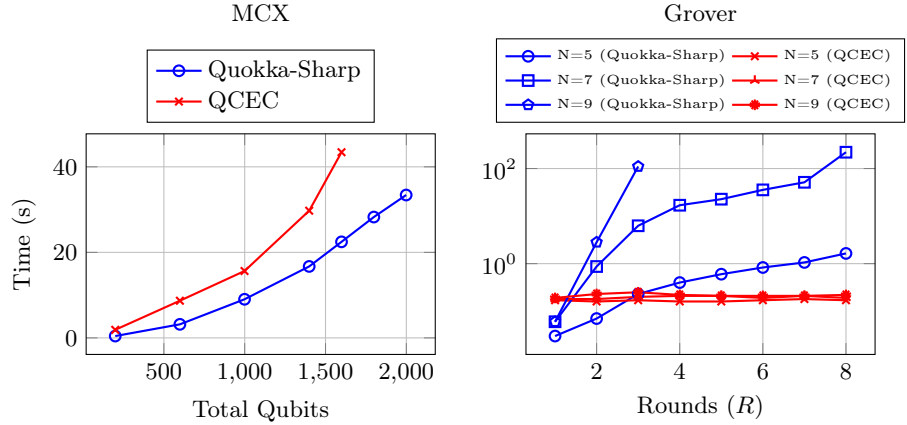

\subsection{Algorithms with Different Parameter Settings}
We further investigate the performance of the two verification backends on
MCX and Grover circuits to distinguish their respective optimal performance ranges.
The vertical axis represents time (seconds).

\cref{fig:2performance} shows complementary regimes.
On MCX (left), \textsc{Quokka-Sharp} scales more smoothly as $N$ increases, while \textsc{QCEC} exhibits rapid growth and eventually hits memory limits at larger sizes.
On Grover (right), \textsc{QCEC} remains nearly constant-time across rounds, whereas \textsc{Quokka-Sharp} grows quickly with the number of rounds and times out at larger configurations.

\subsection{Full Pipeline of Verify--Repair--Reverify}
We further implement Algorithm 2 by extending the existing implementation, which did not originally support this functionality,
and evaluate the full pipeline from diagnosis to local-rotation repair and final re-verification.
For each benchmark, we select a small set of ancilla qubits to cover representative cases: (i) no injection (control), (ii) injected local unitary errors, and (iii) intrinsic nonlocal faults.

\begin{table}[htbp]
\centering
\caption{\textbf{End-to-end verify--repair--reverify results.} ``--'' indicates that the step is not executed.}
\label{tab:full_pipeline}
\resizebox{\textwidth}{!}{
\begin{tabular}{@{}lccc c llllll@{}}
\toprule
\multirow{2}{*}{\textbf{Case}} & \multirow{2}{*}{$N$} & \multirow{2}{*}{\textbf{Depth}} & \multirow{2}{*}{\textbf{Gates}} & \textbf{Inj.} & \textbf{Target} & \multicolumn{2}{c}{\textbf{Diagnosis Phase}} & \textbf{Repair Phase} & \textbf{Validation} & \textbf{Output} \\
\cmidrule(lr){7-8}
 & & & & \textbf{Err.} & \textbf{Qubit} & \textbf{Verify} & \textbf{Entangle Check} & \textbf{Repair Action} & \textbf{ReVerify} & \textbf{(Result)} \\
\midrule

\multirow{2}{*}{Grover} & \multirow{2}{*}{7} & \multirow{2}{*}{55} & \multirow{2}{*}{120} & No & $a_0$ & \stSafe & \dash & \dash & \cmark\ \stSafe & \outSafe \\
 & & & & No & $a_1$ & \stSafe & \dash & \dash & \cmark\ \stSafe & \textit{(Original Safe)} \\
\midrule

\multirow{3}{*}{Grover$_{r_1}$} & \multirow{3}{*}{139} & \multirow{3}{*}{547} & \multirow{3}{*}{1036} & No & $a_0$ & \stSafe & \dash & \dash & \cmark\ \stSafe & \outSafe \\
 & & & & Yes & $a_1$ & \stLogic & \stSep & Apply $R_X$ & \cmark\ \stSafe & \multirow{2}{*}{\textit{(Repaired Safe)}} \\
 & & & & Yes & $a_2$ & \stHybrid & \stSep & Apply $R_Z \cdot R_X \cdot R_Z$ & \cmark\ \stSafe & \\
\midrule

\multirow{3}{*}{MCX} & \multirow{3}{*}{399} & \multirow{3}{*}{792} & \multirow{3}{*}{795} & No & $a_0$ & \stSafe & \dash & \dash & \cmark\ \stSafe & \outSafe \\
 & & & & Yes & $a_1$ & \stPhase & \stSep & Apply $R_Z$ & \cmark\ \stSafe & \multirow{2}{*}{\textit{(Repaired Safe)}} \\
 & & & & Yes & $a_2$ & \stHybrid & \stSep & Apply $R_Z \cdot R_X \cdot R_Z$ & \cmark\ \stSafe & \\
\midrule

\multirow{3}{*}{Adder} & \multirow{3}{*}{399} & \multirow{3}{*}{1588} & \multirow{3}{*}{1591} & No & $a_0$ & \stSafe & \dash & \dash & \cmark\ \stSafe & \multirow{3}{*}{\shortstack{\outUnsafe \\ \textsf{FailList}: $\{a_1\}$}} \\
 & & & & Yes & $a_1$ & \stPhase & \stInsep & \dash & \dash & \\
 & & & & Yes & $a_2$ & \stPhase & \stSep & Apply $R_Z$ & \cmark\ \stSafe & \\
\midrule

\multirow{2}{*}{\shortstack[l]{Bridge\\GHZ}} & \multirow{2}{*}{3999} & \multirow{2}{*}{7997} & \multirow{2}{*}{7998} & No & $a_0$ & \stSafe & \dash & \dash & \cmark\ \stSafe & \outSafe \\
 & & & & Yes & $a_1$ & \stLogic & \stSep & Apply $R_X$ & \cmark\ \stSafe & \textit{(Repaired Safe)} \\
\midrule

\multirow{3}{*}{\shortstack[l]{Reqomp\\-MCX}} & \multirow{3}{*}{599} & \multirow{3}{*}{597} & \multirow{3}{*}{597} & No & $a_0$ & \stPhase & \stInsep & \dash & \dash & \multirow{3}{*}{\shortstack{\outUnsafe \\ \textsf{FailList}: \\ $\{a_0, a_1, a_2\}$}} \\
 & & & & Yes & $a_1$ & \stPhase & \stInsep & \dash & \dash & \\
 & & & & Yes & $a_2$ & \stPhase & \stInsep & \dash & \dash & \\
\midrule

\multirow{3}{*}{\shortstack[l]{Random\\Circuit}} & \multirow{3}{*}{30} & \multirow{3}{*}{23} & \multirow{3}{*}{180} & No & $a_0$ & \stHybrid & \stInsep & \dash & \dash & \multirow{3}{*}{\shortstack{\outUnsafe \\ \textsf{FailList}: \\ $\{a_0, a_1, a_2\}$}} \\
 & & & & Yes & $a_1$ & \stHybrid & \stInsep & \dash & \dash & \\
 & & & & Yes & $a_2$ & \stHybrid & \stInsep & \dash & \dash & \\
\midrule

\multirow{3}{*}{\shortstack[l]{Identity\\Random\\Circuit}} & \multirow{3}{*}{5} & \multirow{3}{*}{48} & \multirow{3}{*}{111} & No & $a_0$ & \stSafe & \dash & \dash & \cmark\ \stSafe & \multirow{3}{*}{\shortstack{\outUnsafe \\ \textsf{FailList}: $\{a_1\}$}} \\
 & & & & Yes & $a_1$ & \stPhase & \stInsep & \dash & \dash & \\
 & & & & Yes & $a_2$ & \stPhase & \stSep & Apply $R_Z$ & \cmark\ \stSafe & \\
\bottomrule
\end{tabular}
}
\end{table}

Table~\ref{tab:full_pipeline} presents the results and illustrates three consistent behaviors. 
(i) When an ancilla is already safe, the pipeline returns True and performs no repair.
(ii) For separable, locally correctable faults (phase/logic/hybrid as detected by the diagnosis), the pipeline applies a local correction and re-verifies to \textsc{True}.
(iii) When the diagnosis indicates an Entanglement, the pipeline does not apply a local repair and adds the qubit to the corresponding failure list \textsf{FailList}.
Finally, the pipeline returns the processed $U_\text{fix}$ and the corresponding \textsf{FailList} (if exists) of failed repairs.
Overall, the pipeline acts as a safety-first mechanism: it repairs exactly the locally invertible cases and otherwise surfaces faults that require global redesign or resynthesis.

\section{Conclusion}\label{sec:conclusion}
We presented an automated, end-to-end framework for verifying and repairing ancilla safety in universal quantum circuits. Our reduction turns dirty-safety on an $m$-qubit ancilla register into $2m$ independent Pauli-$Z$ or $X$ commutativity checks, yielding linear (and parallelizable) backend invocations in $m$. We implemented the full pipeline with interchangeable symbolic backends and validated it on diverse benchmarks, scaling to thousands of qubits. Experiments further show that the verifier yields actionable diagnoses by distinguishing logic errors from phase errors, and that separable ancilla-local faults can be repaired via lightweight single-qubit rotations followed by re-verification, while entanglement-triggered non-local violations are correctly rejected for local repair.

\begin{credits}

\subsubsection*{Acknowledgements}
We thank the anonymous reviewers for their careful and constructive feedback. This work is supported by Quantum Science and Technology-National Science and Technology Major Project under Grant No. 2024ZD0300500, the National Natural Science Foundation of China under Grant No. 62402485, the Youth Innovation Promotion Association, Chinese Academy of Sciences under Grant No. 2023116, the CCF-QuantumCtek Superconducting Quantum Computing Special Cooperation Program under Grant No. CCF-QC2025007, the Young Elite Scientists Sponsorship Program by the China Association for Science and Technology under Grant No. YESS20240449, and the International Partnership Program of the Chinese Academy of Sciences under Grant No. 096GJHZ2025013FN.
JM was supported by the Dutch National Growth Fund, as part of the Quantum Delta NL programme.
WF was supported by the Engineering and Physical Sciences Research Council under Grant EP/X025551/1.

\subsubsection*{Disclosure of Interests}
The authors have no competing interests to declare that are relevant to the content of this article.

\subsubsection*{Data-Availability}
The artifact supporting the experimental results reported in this paper is available on Zenodo at DOI: \href{https://doi.org/10.5281/zenodo.19784589}{10.5281/zenodo.19784589}. It contains the implementation and experimental materials needed to reproduce the reported results, subject to the limitations described in the artifact documentation.

\end{credits}

\bibliographystyle{splncs04}
\bibliography{reference.bib}

\input{appendix}

\end{document}

%% file: appendix.tex
\newpage
\appendix
\renewcommand{\theHsection}{\Alph{section}}
\renewcommand{\theHsubsection}{\Alph{section}.\arabic{subsection}}

\section{The details of \cref{exa:clean-example}}\label{app:detail-example}

To illustrate the proposed verification approach, we analyze the canonical 3-gate circuit shown in \cref{fig:clean_circuit}. This circuit implements a CNOT operation between working qubits $q_0$ and $q_1$ mediated by an ancilla $q_a$.
The system comprises working qubits $W=\{q_0, q_1\}$ and an ancilla $A=\{q_a\}$. The ancilla is initialized in $\ket{0}_a$ and is expected to return to $\ket{0}_a$.

\subsubsection{State Vector Evolution}
First, we verify the clean ancilla safety by tracing the evolution of computational basis states. Let the input state be $\ket{x, y}_\W \ket{0}_a$, where $x, y \in \{0, 1\}$. The evolution of the global state $\ket{\psi}$ proceeds as follows:
\begin{align*}
\ket{\psi_0} &= \ket{x}_{q_0} \ket{y}_{q_1} \ket{0}_{q_a} \\
&\xrightarrow{\text{CNOT}(q_0, q_a)} \ket{x}_{q_0} \ket{y}_{q_1} \ket{x}_{q_a} \\
&\xrightarrow{\text{CNOT}(q_a, q_1)} \ket{x}_{q_0} \ket{x \oplus y}_{q_1} \ket{x}_{q_a} \\
&\xrightarrow{\text{CNOT}(q_0, q_a)} \ket{x}_{q_0} \ket{x \oplus y}_{q_1} \ket{x \oplus x}_{q_a} = \ket{x}_{q_0} \ket{x \oplus y}_{q_1} \ket{0}_{q_a}
\end{align*}

The final state is $\ket{\psi_{final}} = \ket{x, x \oplus y}_\W \otimes \ket{0}_a$. The ancilla is disentangled and restored to $\ket{0}_a$ for all basis inputs, satisfying the clean ancilla safety definition.

\subsubsection{Subspace Invariance}
We examine the evolution of the linear subspace spanned by valid clean states.
Let $\mathcal{V}_{clean} = \text{span}\{ \ket{q_0 q_1}\ket{0}_a \}$ be the subspace where the ancilla is in $\ket{0}$.
The basis vectors for this subspace are $\{ \ket{000}, \ket{010}, \ket{100}, \ket{110} \}$.

\paragraph{1. Detailed Trace of a Representative Basis Vector}
Let us track the evolution of the basis vector $\ket{100}$ (where $q_0=1, q_1=0, q_a=0$). This state is non-trivial because the control qubit $q_0$ is active.

\begin{itemize}
    \item \textbf{Input State:} 
    \[ \ket{\psi_0} = \ket{1}_{q_0}\ket{0}_{q_1}\ket{0}_{q_a} \in \mathcal{V}_{clean} \]
    
    \item \textbf{After Gate 1 ($q_0 \to q_a$):}
    Since $q_0=1$, the ancilla flips to $\ket{1}$.
    \[ \ket{\psi_1} = \ket{1}_{q_0}\ket{0}_{q_1}\ket{\mathbf{1}}_{q_a} \]
    \textit{Observation:} The state vector has left the subspace $\mathcal{V}_{clean}$. It is now in the orthogonal subspace $\mathcal{V}_{dirty}$ (where $q_a=1$).
    
    \item \textbf{After Gate 2 ($q_a \to q_1$):}
    Since $q_a=1$ (dirty), it activates the second gate, flipping $q_1$.
    \[ \ket{\psi_2} = \ket{1}_{q_0}\ket{\mathbf{1}}_{q_1}\ket{1}_{q_a} \]
    \textit{Observation:} The computation is performed correctly, but the state remains outside $\mathcal{V}_{clean}$.
    
    \item \textbf{After Gate 3 ($q_0 \to q_a$):}
    Since $q_0=1$ still holds, the ancilla flips again ($1 \oplus 1 = 0$).
    \[ \ket{\psi_{final}} = \ket{1}_{q_0}\ket{1}_{q_1}\ket{\mathbf{0}}_{q_a} \in \mathcal{V}_{clean} \]
    \textit{Conclusion:} The final vector returns to the clean subspace.
\end{itemize}

\paragraph{2. Summary of Full Basis Evolution}
Similarly, we can verify the evolution for the remaining basis vectors of $\mathcal{V}_{clean}$. The results are summarized below:

\[
\begin{aligned}
    \hat{U} \ket{000} &= \ket{000} \in \mathcal{V} \qquad
    \hat{U} \ket{010} = \ket{010} \in \mathcal{V} \\
    \hat{U} \ket{100} &= \ket{110} \in \mathcal{V} \qquad
    \hat{U} \ket{110} = \ket{100} \in \mathcal{V}
\end{aligned}
\]

\paragraph{Geometric Conclusion}
Although the intermediate states (like $\ket{\psi_1}$ and $\ket{\psi_2}$) may rotate out of $\mathcal{V}_{clean}$ into the dirty subspace, the unitary $U$ maps every basis vector of $\mathcal{V}_{clean}$ back into $\mathcal{V}_{clean}$.
Since the operator is linear, this implies:
\[ U(\mathcal{V}_{clean}) = \mathcal{V}_{clean} \]
Thus, the clean ancilla safety is geometrically satisfied.

\subsubsection{Commutativity of Reflector}
Finally, we apply our reduction theorem by verifying the commutativity of the unitary $U$ with the Householder reflection operator $Q_{\mathcal{V}}$.
For the clean subspace determined by $\ket{0}_a$, the reflection operator is derived as:
\[
Q_{\mathcal{V}} = 2 (I_\W \otimes \ket{0}\bra{0}_a) - I_{W \cup A} = I_\W \otimes Z_a.
\]
We verify $[U, Q_{\mathcal{V}}] = 0$ by checking if the operator $Q_{\mathcal{V}}$ remains invariant under the conjugation of each gate in the sequence $U = U_3 U_2 U_1$. We denote the evolving operator as $Q^{(t)}$.

\textbf{Initialization:} $Q^{(0)} = I_{q_0} \otimes I_{q_1} \otimes Z_{q_a}$.

\textbf{Step 1: } $U_1 = \text{CNOT}(q_0, q_a)$.
Using the propagation rule $\text{CNOT}_{c,t} (I_c \otimes Z_t) \text{CNOT}_{c,t} = Z_c \otimes Z_t$:
\[
Q^{(1)} = U_1 Q^{(0)} U_1^\dagger = Z_{q_0} \otimes I_{q_1} \otimes Z_{q_a}.
\]

\textbf{Step 2: } $U_2 = \text{CNOT}(q_a, q_1)$.
Here $q_a$ is the control qubit. Since the control qubit's Z-operator commutes with CNOT ($\text{CNOT}_{c,t} Z_c = Z_c \text{CNOT}_{c,t}$):
\[
Q^{(2)} = U_2 Q^{(1)} U_2^\dagger = Z_{q_0} \otimes I_{q_1} \otimes Z_{q_a}.
\]

\textbf{Step 3: } $U_3 = \text{CNOT}(q_0, q_a)$.
We apply the conjugation to $Z_{q_0} \otimes Z_{q_a}$. Note that $\text{CNOT}_{c,t}$ maps $Z_c \otimes Z_t$ to $I_c \otimes Z_t$:
\[
Q^{(3)} = U_3 Q^{(2)} U_3^\dagger = I_{q_0} \otimes I_{q_1} \otimes Z_{q_a}.
\]

\textbf{Conclusion:}
Since $Q^{(3)} = Q^{(0)} = I_\W \otimes Z_a$, we have $U Q_{\mathcal{V}} U^\dagger = Q_{\mathcal{V}}$, which implies $[U, Q_{\mathcal{V}}] = 0$. By \cref{thm:clean-reduction}, the circuit is clean safe.

\subsection{Motivating Example: Bridge based GHZ Preparation}
\label{subsec:ghz_bridge}

To demonstrate the capabilities of our verification tool, we introduce a parameterized family of GHZ state preparation circuits designed for architectures with limited connectivity. 

\subsubsection*{The Bridge Gate Concept}
In Superconducting Quantum Processors (SQPs), qubit connectivity is typically restricted to a sparse coupling map (e.g., nearest-neighbor). Executing a CNOT between non-adjacent qubits usually requires inserting a chain of SWAP gates, which incurs high overhead (3 CNOTs per SWAP). 
To mitigate this, Itoko et al.~\cite{Itoko2020} proposed the \textit{Bridge Gate} technique. Instead of swapping states, a bridge gate utilizes an intermediate qubit (ancilla) to mediate the interaction. 
As shown in Fig.~\ref{fig:clean_bridge ghz}, a logical CNOT between control $q_c$ and target $q_t$ is realized via an ancilla $a$.

\subsubsection*{Benchmark Design: Cascaded GHZ Generation}
We construct an $n$-qubit GHZ state ($|0\dots0\rangle + |1\dots1\rangle$) by cascading $n-1$ bridge gates. Each step performs a logical CNOT from $q_i$ to $q_{i+1}$ via an ancilla. 
To evaluate robustness and resource efficiency, we define four variants of this benchmark based on two dimensions: \textit{Safety Property} and \textit{Ancilla Allocation}.

\paragraph{1. Safety Property (Clean vs. Dirty).}
This dimension defines the assumption on the initial state of the ancilla and the gate construction used:
\begin{itemize}
    \item \textbf{Clean Bridge (3-CNOTs):} As Fig~\ref{fig:clean_bridge ghz} shows, this circuit assumes the ancilla starts in the ground state $|0\rangle$. The sequence is $CX(q_c, a) \rightarrow CX(a, q_t) \rightarrow CX(q_c, a)$. 
    \textit{Verification Challenge:} If the ancilla is initialized to $|0\rangle$, it returns to $|0\rangle$ and performs the correct logic. However, if the ancilla is dirty (e.g., $|1\rangle$), the target qubit suffers a bit-flip error. This requires verifying \textit{Clean Safety}.
    
    \item \textbf{Dirty Bridge (4-CNOTs):} As Fig~\ref{fig:dirty_bridge ghz} shows, this circuit makes no assumption on the ancilla's initial state. It appends a fourth gate, $CX(a, q_t)$, to the Clean Bridge sequence. 
    \textit{Verification Challenge:} This construction guarantees that the ancilla is restored and the logical CNOT is correct regardless of the ancilla's initial state (even if entangled). This requires verifying \textit{Dirty Safety}.
\end{itemize}

\begin{figure}[ht]
    \begin{center}
        \begin{quantikz}
        \lstick{$q_0$}  & \gate H    & \ctrl{2} & \qw      & \ctrl{2}    &\qw \\
        \lstick{$q_1$}  & \qw        & \qw      & \targ{}  & \qw         &\qw \\
        \lstick{$q_a$}  & \qw        & \targ{}  & \ctrl{-1}& \targ{}     &\qw
        \end{quantikz}
        \caption{A Bridge CNOT circuit implementing a mediated operation. The ancilla $q_a$ is used to transfer information and is uncomputed at the end to satisfy clean safety.}
        \label{fig:clean_bridge ghz}
    \end{center}
\end{figure}
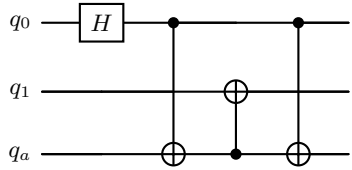

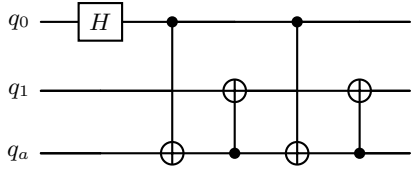
\begin{figure}[ht]
    \begin{center}
        \begin{quantikz}
        \lstick{$q_0$}  & \gate H    & \ctrl{2} & \qw      & \ctrl{2} & \qw    &\qw \\
        \lstick{$q_1$}  & \qw        & \qw      & \targ{}  & \qw      & \targ{}  &\qw \\
        \lstick{$q_a$}  & \qw        & \targ{}  & \ctrl{-1}& \targ{}  & \ctrl{-1} &\qw
        \end{quantikz}
        \caption{A Bridge CNOT circuit implemented by a dirty ancilla between $q_0$ and $q_1$. The ancilla $q_a$ is used to transfer information and is uncomputed at the end to satisfy dirty safety.}
        \label{fig:dirty_bridge ghz}
    \end{center}
\end{figure}

\paragraph{2. Ancilla Allocation (Distinct vs. Shared).}
This dimension defines how ancilla resources are managed across the cascade:
\begin{itemize}
    \item \textbf{Distinct Ancillae (Linear):} Each logical CNOT uses a fresh, dedicated ancilla. For an $n$-qubit GHZ state, this requires $n-1$ ancilla qubits. This represents a resource-rich scenario where spatial locality is prioritized. The example circuit is shown in Fig~\ref{fig:Double dirty ancilla ghz dirty}.
    
    \item \textbf{Shared Ancilla (Reused):} A single ancilla is reused sequentially for all $n-1$ logical operations. This represents a resource-constrained scenario. 
    \textit{Verification Challenge:} This introduces complex temporal dependencies. A failure to restore the ancilla in step $i$ will propagate errors to all subsequent steps $j > i$, making verification significantly harder. The example circuit is shown is Fig~\ref{fig:Double 1 dirty ancilla ghz dirty}.
\end{itemize}

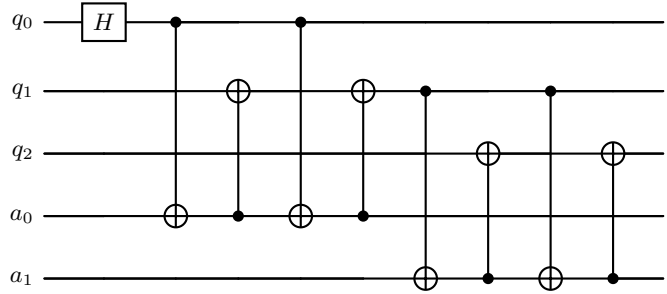
\begin{figure}[ht]
    \centering
    \begin{quantikz}
\lstick{$q_0$}& \gate H     & \ctrl{3} & \qw      & \ctrl{3} & \qw      & \qw      & \qw      & \qw      & \qw      &  \qw\\
\lstick{$q_1$}& \qw      & \qw      & \targ{}  & \qw      & \targ{}  & \ctrl{3} & \qw      & \ctrl{3} & \qw      & \qw\\
\lstick{$q_2$}& \qw      & \qw      & \qw      & \qw      & \qw      & \qw      & \targ{}  & \qw      & \targ{}  & \qw\\  
\lstick{$a_0$}& \qw      & \targ{}  & \ctrl{-2}& \targ{}  & \ctrl{-2}& \qw      & \qw      & \qw      & \qw      &  \qw\\
\lstick{$a_1$}& \qw      & \qw      & \qw      & \qw      & \qw      & \targ{}  & \ctrl{-2}& \targ{}  & \ctrl{-2}  & \qw
\end{quantikz}
    \caption{A cascaded GHZ quantum circuit design utilizing two dirty ancilla qubits ($a_0, a_1$) to mediate interactions between data qubits.}
    \label{fig:Double dirty ancilla ghz dirty}
\end{figure}

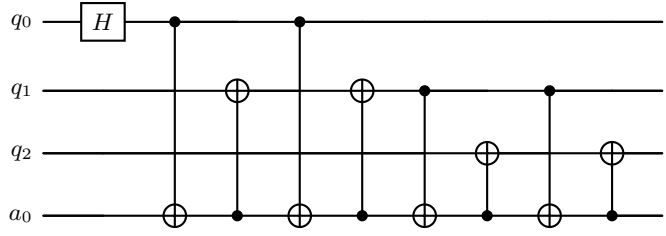
\begin{figure}
    \centering
    \begin{quantikz}
\lstick{$q_0$}& \gate H    & \ctrl{3} & \qw      & \ctrl{3} & \qw      & \qw      & \qw      & \qw         & \qw         &  \qw\\
\lstick{$q_1$}& \qw        & \qw      & \targ{}  & \qw      & \targ{}  & \ctrl{2}    & \qw      & \ctrl{2} & \qw      & \qw\\
\lstick{$q_2$}& \qw        & \qw      & \qw      & \qw      & \qw      & \qw          & \targ{}  & \qw      & \targ{}  & \qw\\      
\lstick{$a_0$}& \qw        & \targ{}  & \ctrl{-2}& \targ{}  & \ctrl{-2}& \targ{}      & \ctrl{-1}         & \targ{}      & \ctrl{-1}      &  \qw
\end{quantikz}
    \caption{A cascaded GHZ quantum circuit design utilizing one dirty ancilla qubit $a_0$ to mediate interactions between data qubits.}
    \label{fig:Double 1 dirty ancilla ghz dirty}
\end{figure}

These variants provide a comprehensive suite for benchmarking quantum verification tools, ranging from simple clean checks to complex, interference-prone dirty safety verification. 

\section{Proof}\label{app:proof}

\begin{proof}[Proof of \cref{lem:clean-subspace}]
Let $\mathcal{V} = \mathcal{H}_\W \otimes \mathrm{span}\{\ket{0}_\A\}$.

($\Rightarrow$)
Assume $U$ is clean-safe on $\A$.
By definition, for any $\ket{\psi}_\W \in \mathcal{H}_\W$,
\[
U(\ket{\psi}_\W \otimes \ket{0}_\A) = \ket{\psi'}_\W \otimes \ket{0}_\A \in \mathcal{V}.
\]
By linearity, for any $\ket{v} \in \mathcal{V}$, $U\ket{v} \in \mathcal{V}$. Thus,
\[
U(\mathcal{V}) \subseteq \mathcal{V}.
\]
Since $U$ is unitary, it is an isometry and preserves subspace dimensions:
\[
\dim(U(\mathcal{V})) = \dim(\mathcal{V}).
\]
For finite-dimensional subspaces, inclusion ($U(\mathcal{V}) \subseteq \mathcal{V}$) and dimension equality imply equality:
\[
U(\mathcal{V}) = \mathcal{V}.
\]

($\Leftarrow$)
Assume $U(\mathcal{V}) = \mathcal{V}$.
For any $\ket{\psi}_\W \in \mathcal{H}_\W$, let $\ket{\phi} = \ket{\psi}_\W \otimes \ket{0}_\A$. Clearly $\ket{\phi} \in \mathcal{V}$.
The assumption implies $U\ket{\phi} \in \mathcal{V}$.
By the definition of $\mathcal{V}$, any vector in $\mathcal{V}$ takes the form $\ket{\xi}_\W \otimes \ket{0}_\A$. Hence, there exists $\ket{\psi'}_\W \in \mathcal{H}_\W$ such that
\[
U(\ket{\psi}_\W \otimes \ket{0}_\A) = \ket{\psi'}_\W \otimes \ket{0}_\A,
\]
which satisfies the definition of clean safety.
\end{proof}

\begin{proof}[Proof of \cref{lemma:clean-commute}]
Recall that $\mathcal{V} = \mathcal{H}_\W \otimes \mathrm{span}\{\ket{0}_\A\}$. Let $P_{\mathcal{V}} = I_\W \otimes \ket{0}\!\bra{0}_\A$ be the orthogonal projector onto $\mathcal{V}$.
The operator $Q_{\mathcal{V}}$ is defined as:
\[
Q_{\mathcal{V}} = 2 P_{\mathcal{V}} - I.
\]
Since the identity operator $I$ commutes with any operator, the commutator relation simplifies to:
\[
[U, Q_{\mathcal{V}}] = [U, 2P_{\mathcal{V}} - I] = 2[U, P_{\mathcal{V}}].
\]
Thus, $[U, Q_{\mathcal{V}}] = 0 \iff [U, P_{\mathcal{V}}] = 0$. The proof reduces to showing $U(\mathcal{V})=\mathcal{V} \iff [U, P_{\mathcal{V}}] = 0$.

($\Rightarrow$)
Assume $U(\mathcal{V}) = \mathcal{V}$.
Consider the operator $R = U P_{\mathcal{V}} U^\dagger$. Since $U$ is unitary, $R$ is an orthogonal projector.
For any $\ket{u} \in \mathcal{H}$, let $\ket{v} = P_{\mathcal{V}} (U^\dagger \ket{u})$. By definition, $\ket{v} \in \mathcal{V}$. Then,
\[
R \ket{u} = U \ket{v} \in U(\mathcal{V}) = \mathcal{V}.
\]
Thus $\mathrm{Im}(R) = \mathcal{V}$. Since the orthogonal projector onto a subspace is unique,
\[
U P_{\mathcal{V}} U^\dagger = P_{\mathcal{V}}.
\]
Right-multiplying by $U$ yields $U P_{\mathcal{V}} = P_{\mathcal{V}} U$, i.e., $[U, P_{\mathcal{V}}] = 0$.

($\Leftarrow$)
Assume $[U, P_{\mathcal{V}}] = 0$.
For any $\ket{v} \in \mathcal{V}$, we have $P_{\mathcal{V}}\ket{v} = \ket{v}$.
Applying the projection to the transformed vector $U\ket{v}$:
\[
P_{\mathcal{V}} (U \ket{v}) = U (P_{\mathcal{V}} \ket{v}) = U \ket{v}.
\]
Since $P_{\mathcal{V}}$ acts as the identity only on $\mathcal{V}$, this implies $U\ket{v} \in \mathcal{V}$.
Thus, $U(\mathcal{V}) \subseteq \mathcal{V}$.
Since $U$ is an isometry on a finite-dimensional space, inclusion implies equality:
\[
U(\mathcal{V}) = \mathcal{V}. \qedhere
\]
\end{proof}

\begin{proof}[Proof of \cref{lem:dirty-multi-single}]
($\Rightarrow$)
Assume $U$ is dirty safe with respect to the full ancilla register $\A = \{a_1, \dots, a_m\}$.
By definition, there exists a unitary $V$ on $\mathcal{H}_\W$ such that $U = V \otimes I_\A$, which can be expanded as:
\[
U = V \otimes I_{a_1} \otimes \cdots \otimes I_{a_m}.
\]
Then, for any $j \in \{1, \dots, m\}$, any state $\ket{\varphi}_{a_j} \in \mathcal{H}_{a_j}$, and any state $\ket{\psi}_{\W \cup (\A \setminus \{a_j\})} \in \mathcal{H}_{\W \cup (\A \setminus \{a_j\})}$, we have
\[
U (\ket{\psi}_{\W \cup (\A \setminus \{a_j\})} \otimes \ket{\varphi}_{a_j})
= (V \otimes I_{\A \setminus \{a_j\}}) \ket{\psi}_{\W \cup (\A \setminus \{a_j\})} \otimes \ket{\varphi}_{a_j},
\]
which shows $U$ is dirty safe with respect to $a_j$ alone.

($\Leftarrow$)
Assume for every $a_j \in \A$, $U$ is dirty safe on $a_j$. That is,
\[
\forall j \in \{1, \dots, m\},\ \exists V_j \quad \text{on} \quad \mathcal{H}_{\W\cup (\A\setminus \{a_j\})},
\]
such that
\[
U = V_j \otimes I_{a_j}.
\]
Since this holds for all $j = 1, \dots, m$, $U$ acts as the identity on every ancilla qubit independently.
Hence, $U$ acts as the identity on the entire ancilla register $\A$:
\[
U = V \otimes I_\A
\]
for some unitary $V$ on $\mathcal{H}_\W$.
\end{proof}

\begin{proof}[Proof of \cref{thm:clean-reduction}]
This result follows directly from the chain of equivalences established in Lemma~\ref{lem:clean-subspace} and Lemma~\ref{lemma:clean-commute}.

First, by Lemma~\ref{lem:clean-subspace}, clean safety is equivalent to the invariance of the clean subspace $\mathcal{V}$:
\[
U \text{ is clean-safe on } \A \iff U(\mathcal{V}) = \mathcal{V}.
\]
Second, by Lemma~\ref{lemma:clean-commute}, the invariance of subspace $\mathcal{V}$ is equivalent to the commutation of $U$ with the operator $Q_{\mathcal{V}}$:
\[
U(\mathcal{V}) = \mathcal{V} \iff [U, Q_{\mathcal{V}}] = 0.
\]
Combining these two equivalences yields:
\[
U \text{ is clean-safe on } \A \iff [U, Q_{\mathcal{V}}] = 0. \qedhere
\]
\end{proof}

\begin{proof}[Proof of \cref{thm:clean-reduction-gen}]

Let $\mathcal{V}_{\varphi} = \mathcal{H}_\W \otimes \mathrm{span}\{\ket{\varphi}_\A\}$ and $P_{\mathcal{V}_{\varphi}} = I_\W \otimes \ket{\varphi}\!\bra{\varphi}_\A$. The proof proceeds by establishing a chain of equivalences.

\noindent\textbf{1. Clean safety $\iff$ Subspace invariance}
\begin{itemize}
    \item By Definition, $U$ is clean-safe for $\ket{\varphi}_\A$ implies that for all $\ket{v} \in \mathcal{V}_{\varphi}$, $U\ket{v} \in \mathcal{V}_{\varphi}$. Thus, $U(\mathcal{V}_{\varphi}) \subseteq \mathcal{V}_{\varphi}$.
    \item Since $U$ is unitary (isometry) and $\dim(\mathcal{V}_{\varphi}) < \infty$:
    \[
    U(\mathcal{V}_{\varphi}) \subseteq \mathcal{V}_{\varphi} \implies U(\mathcal{V}_{\varphi}) = \mathcal{V}_{\varphi}.
    \]
\end{itemize}

\noindent\textbf{2. Subspace invariance $\iff$ Commutativity}
\begin{itemize}
    \item Assume $U(\mathcal{V}_{\varphi}) = \mathcal{V}_{\varphi}$. Then the projection onto the subspace is invariant under conjugation by $U$:
    \[
    U P_{\mathcal{V}_{\varphi}} U^\dagger = P_{\mathcal{V}_{\varphi}} \iff U P_{\mathcal{V}_{\varphi}} = P_{\mathcal{V}_{\varphi}} U \iff [U, P_{\mathcal{V}_{\varphi}}] = 0.
    \]
    \item Since $Q_{\mathcal{V}_{\varphi}} = 2P_{\mathcal{V}_{\varphi}} - I$, and $[U, I] = 0$:
    \[
    [U, P_{\mathcal{V}_{\varphi}}] = 0 \iff [U, 2P_{\mathcal{V}_{\varphi}} - I] = 0 \iff [U, Q_{\mathcal{V}_{\varphi}}] = 0.
    \]
\end{itemize}
Combining \textbf{1} and \textbf{2}, we conclude:
\[
U \text{ is clean-safe for } \ket{\varphi}_\A \iff [U, Q_{\mathcal{V}_{\varphi}}] = 0. \qedhere
\]
\end{proof}

\begin{proof}[Proof of \cref{thm:dirty-reduction}]
($\Rightarrow$)
Assume $U$ is dirty-safe on $\A$.
By Lemma~\ref{lem:dirty-multi-single}, $U$ acts as the identity on the ancilla register $\A$ (up to a unitary on $\W$). That is, $U = V_\W \otimes I_\A$.
Since the identity operator $I_\A$ commutes with any operator on $\mathcal{H}_\A$, $U$ commutes with $Q_{(a_i, \varphi)}$ for any $a_i$ and any $\varphi$.

($\Leftarrow$)
Assume the commutativity condition holds:
\[
\forall\,a_i\in \A,\ \forall\,\ket{\varphi}\in\{\ket{\varphi_1},\ket{\varphi_2}\}:\quad [U,\,Q_{(a_i,\ket\varphi)}]=0.
\]
By Lemma~\ref{lem:dirty-multi-single}, it suffices to prove that $U$ acts as the identity on each individual ancilla qubit $a_i$ (i.e., $U$ is dirty-safe on each $a_i$).

Fix an arbitrary ancilla qubit $a_i$.
By Theorem~\ref{thm:clean-reduction-gen}, the commutativity conditions imply that $U$ is clean-safe for both $\ket{\varphi_1}_{a_i}$ and $\ket{\varphi_2}_{a_i}$.
Let us decompose the action of $U$ relative to the basis aligned with $\ket{\varphi_1}$. Without loss of generality, identify $\ket{\varphi_1}$ with the basis state $\ket{0}$ and let $\ket{0}^\perp = \ket{1}$.
Since $U$ is clean-safe for $\ket{0}$, it must be block-diagonal with respect to this basis:
\[
U = V_0 \otimes \ket{0}\!\bra{0}_{a_i} + V_1 \otimes \ket{1}\!\bra{1}_{a_i},
\]
where $V_0, V_1$ are unitaries on the remaining qubits.

Now consider the second state $\ket{\varphi_2}$. Since $0 < \left| \braket{\varphi_1}{\varphi_2} \right| < 1$, we can write $\ket{\varphi_2} = \alpha\ket{0} + \beta\ket{1}$ with non-zero coefficients $\alpha, \beta \in \mathbb{C} \setminus \{0\}$.
Applying $U$ to an arbitrary state $\ket{\psi}$ on the working register tensored with $\ket{\varphi_2}$:
\[
U(\ket{\psi}\otimes\ket{\varphi_2}) = \alpha (V_0\ket{\psi})\otimes\ket{0} + \beta (V_1\ket{\psi})\otimes\ket{1}.
\]
Since $U$ is also clean-safe for $\ket{\varphi_2}$, the output must separate as $\ket{\psi'}\otimes\ket{\varphi_2} = \ket{\psi'}\otimes(\alpha\ket{0} + \beta\ket{1})$.
Matching the coefficients for $\ket{0}$ and $\ket{1}$ implies:
\[
V_0 \ket{\psi} = \ket{\psi'} \quad \text{and} \quad V_1 \ket{\psi} = \ket{\psi'}.
\]
Thus $V_0 = V_1$.
Consequently, $U = V_0 \otimes (\ket{0}\!\bra{0} + \ket{1}\!\bra{1}) = V_0 \otimes I_{a_i}$.
This proves that $U$ acts trivially on the ancilla qubit $a_i$. Since this holds for all $a_i \in A$, $U$ is dirty-safe on $A$.
\end{proof}

\section{Implementation Details of Repair Components}
\label{app:implementation}

\subsection{Entanglement Check via Partial Trace}
To efficiently verify this, we reduce the problem to a purity computation defined by a partial trace followed by a trace calculation.

For $|\phi\rangle\in\{|0\rangle,|+\rangle\}$, let
\[
Q'_{(a,|\phi\rangle)}
=
UQ_{(a,|\phi\rangle)}U^\dagger,
\]
and denote all qubits other than $a$ by
$R=W\cup(A\setminus\{a\})$.
A local repair on $a$ is possible only if
\[
Q'_{(a,|\phi\rangle)}
=
I_R\otimes O_{a,\phi}
\]
for some single-qubit operator $O_{a,\phi}$.

To check this condition, we compute the normalized partial trace
\[
O_{a,\phi}
=
\frac{1}{d_R}
\operatorname{Tr}_R
\left(
Q'_{(a,|\phi\rangle)}
\right),
\qquad
d_R=\dim(\mathcal H_R).
\]
Since $Q'_{(a,|\phi\rangle)}$ is a Hermitian unitary, we have
\[
p_{a,\phi}
:=
\frac{1}{2}\operatorname{Tr}(O_{a,\phi}^2)
\le 1,
\]
where equality holds if and only if
$Q'_{(a,|\phi\rangle)}=I_R\otimes O_{a,\phi}$.

Therefore, Algorithm~2 proceeds with local repair only when
\[
p_{a,0}=p_{a,+}=1.
\]
Otherwise, the evolved witness contains non-local components and the
target ancilla is added to \textsc{FailList}.
In practice, the equality is checked up to a numerical tolerance.

\subsection{Angle Calculation}
Recall $Q_{(a,\ket0)} = I_{\W\cup\A/\{a\}}\otimes Z$ and $Q_{(a,\ket+)} = I_{\W\cup\A/\{a\}}\otimes X$. $Q_{(a,\ket0)}' = UQ_{(a,\ket0)}U^\dagger$ and $Q_{(a,\ket+)}' = UQ_{(a,\ket+)}U^\dagger$. 
Since the error is separable, Then there exists $O_{a,0}$ and $O_{a,+}$ that are both unitary, satisfying \[
Q_{(a,\ket0)}' = I_{\W\cup\A/\{a\}}\otimes O_{a,0}
\]
and \[
Q_{(a,\ket+)}' = I_{\W\cup\A/\{a\}}\otimes O_{a,+}.
\]

\subsubsection{Both error occurs}
If \texttt{LogicError} and \texttt{PhaseError} both occur, then we find $\theta_{z1}$, $\theta_{x}$, $\theta_{z2}$
which satisfy:
\[
R_Z(\theta_{z1})R_X(\theta_x)R_Z(\theta_{z2}) \cdot O_{a,0}\cdot R_Z(\theta_{z1})R_X(\theta_x)R_Z(\theta_{z2})^\dagger  = Z
\]
and:
\[
R_Z(\theta_{z1})R_X(\theta_x)R_Z(\theta_{z2}) \cdot O_{a,+}\cdot R_Z(\theta_{z1})R_X(\theta_x)R_Z(\theta_{z2})^\dagger  = X.
\]

\paragraph{Existence and Uniqueness of the Solution.}
The problem of finding the repair angles $(\theta_{z1}, \theta_x, \theta_{z2})$ is equivalent to constructing the inverse unitary $U_{\mathrm{err}}^\dagger$ such that the corrected state aligns with the computational basis.
\textbf{Existence} is guaranteed by the Euler decomposition theorem, which states that any element of the special unitary group $\mathrm{SU}(2)$ can be parameterized by three real rotations $R_Z(\gamma)R_X(\beta)R_Z(\alpha)$. Since the accumulated error $U_{\mathrm{err}}$ is a single-qubit unitary operator, there strictly exists a set of angles that satisfies the alignment constraints for both the $|0\rangle$ and $|+\rangle$ states simultaneously.

Regarding \textbf{Uniqueness}, the standard Euler decomposition is generally unique up to coordinate singularities (Gimbal lock) and periodicities. 
However, we enforce a deterministic canonical solution selected by our convention within the principal interval $(-\pi, \pi]$ through our constructive algorithm.
First, the periodicity ambiguity is resolved by the use of the two-argument arctangent function ($\mathrm{arctan2}$), which maps all inputs deterministically to the principal range.
Second, the potential degeneracy at the poles (where $\theta_x = 0$ or $\pi$, causing $\theta_{z1}$ and $\theta_{z2}$ to become linearly dependent) is resolved by the sequential nature of our derivation. 
Specifically, $\theta_{z1}$ is uniquely determined by the projection of the $Z$-vector onto the $Y$-$Z$ plane; once fixed, $\theta_x$ and $\theta_{z2}$ are subsequently determined by the remaining deviations.
Thus, for any physically valid error unitary, our analytic procedure yields exactly one repair sequence.

\paragraph{\texttt{LogicError} only or \texttt{PhaseError} only}
When \texttt{LogicError} is detected only, then we find $\theta\in (-\pi,\pi]$ satisfying:
\[
R_X(\theta)O_{a,0}R_X(\theta)^\dagger = Z.
\]

When \texttt{PhaseError} is detected only, then we find $\theta\in (-\pi,\pi]$ satisfying:
\[
R_Z(\theta)O_{a,+}R_Z(\theta)^\dagger = X.
\]

The existence and uniqueness of the solution $\theta$ are guaranteed by the algebraic constraints imposed by the diagnosis phase.
For a pure \texttt{LogicError}, the commutativity check $[U_{\mathrm{err}}, X]=0$ restricts the error unitary to the subgroup generated by the Pauli-$X$ operator.
Geometrically, this confines the Bloch vector of the ancilla to the $Y$-$Z$ plane.
Consequently, the optimization problem reduces to finding a unique rotation angle on a unit circle that aligns the vector with the north pole $|0\rangle$, which is analytically solvable via $\theta = \arctan2(r_y, r_z)$.
Similarly, a pure \texttt{PhaseError} implies $[U_{\mathrm{err}}, Z]=0$, constraining the state to the equatorial $X$-$Y$ plane.
The solution is thus the unique azimuthal angle required to restore the $|+\rangle$ state, derived as $\theta = \arctan2(-r_y, r_x)$.
In both cases, the dimensionality reduction from the full Bloch sphere to a single great circle ensures a one-to-one mapping between the error syndrome and the repair angle within the principal branch.

\section{More Experiment Results about Soundness and Scalability of Verification}
\label{app:full_data_unified_final}

\begin{table*}[htbp]
\centering
\caption{\textbf{Comparison of verification efficiency and results.} Time is in seconds. ``--'' indicates the step is skipped due to timeout/memout.}
\label{tab:full_data_unified_final}
\renewcommand{\arraystretch}{0.9} 
\setlength{\tabcolsep}{4pt}

\resizebox{\textwidth}{!}{
\begin{tabular}{@{}l cc rr rr rr@{}}
\toprule
\multirow{2}{*}{\textbf{Benchmark}} & \multirow{2}{*}{\boldmath$N$} & \multirow{2}{*}{\textbf{Depth}} & \multicolumn{2}{c}{\textbf{Matrix}} & \multicolumn{2}{c}{\textbf{Quokka-Sharp}} & \multicolumn{2}{c}{\textbf{QCEC}} \\
\cmidrule(lr){4-5} \cmidrule(lr){6-7} \cmidrule(lr){8-9}
 & & & \textbf{Time} & \textbf{Output} & \textbf{Time} & \textbf{Output} & \textbf{Time} & \textbf{Output} \\
\midrule

\multirow{6}{*}{Grover$_F$}
 & 7   & 55      & 0.1    & \stSafe & 6.4    & \stSafe & 0.2  & \stSafe \\
 & 9   & 99      & 1.0    & \stSafe & Timeout & \dash   & 0.3  & \stSafe \\
 & 11  & 195     & 38.3   & \stSafe & Timeout & \dash   & 0.3  & \stSafe \\
 & 19  & 1603    & Memout & \dash   & Timeout & \dash   & 0.4  & \stSafe \\
 & 29  & 14771   & Memout & \dash   & Timeout & \dash   & 1.9  & \stSafe \\
 & 35  & 51459   & Memout & \dash   & Timeout & \dash   & 8.0  & \stSafe \\
 & 39  & 115779   & Memout & \dash   & Timeout & \dash   & 25.9  & \stSafe \\
\midrule

\multirow{8}{*}{Grover$_{r1}$}
 & 99  & 387  & Memout & \dash & 7.0   & \stSafe & 1.0  & \stSafe \\
 & 179 & 707  & Memout & \dash & 21.8  & \stSafe & 2.0  & \stSafe \\
 & 239 & 947  & Memout & \dash & 41.9  & \stSafe & 3.1  & \stSafe \\
 & 319 & 1267 & Memout & \dash & 68.3  & \stSafe & 5.5  & \stSafe \\
 & 399 & 1587 & Memout & \dash & 112.2 & \stSafe & 6.6  & \stSafe \\
 & 479 & 1907 & Memout & \dash & 162.4 & \stSafe & 8.6  & \stSafe \\
 & 559 & 2227 & Memout & \dash & 223.8 & \stSafe & 15.6 & \stSafe \\
 & 699 & 2787 & Memout & \dash & 370.0 & \stSafe & 19.5 & \stSafe \\
\midrule

\multirow{11}{*}{MCX}
 & 199  & 392    & Memout & \dash & 0.6   & \stSafe & 2.2    & \stSafe \\
 & 599  & 1192   & Memout & \dash & 3.7   & \stSafe & 8.7    & \stSafe \\
 & 999  & 1992   & Memout & \dash & 10.0  & \stSafe & 18.7   & \stSafe \\ 
 & 1399 & 2792   & Memout & \dash & 19.2  & \stSafe & 35.3   & \stSafe \\
 & 1599 & 3192   & Memout & \dash & 28.0  & \stSafe & 50.2   & \stSafe \\ 
 & 1799 & 3592   & Memout & \dash & 31.8  & \stSafe & 57.7   & \stSafe \\ 
 & 1999 & 3992   & Memout & \dash & 38.9  & \stSafe & 63.1   & \stSafe \\
 & 3999 & 7992   & Memout & \dash & 174.0 & \stSafe & 514.4  & \stSafe \\
 & 5999 & 11992  & Memout & \dash & 414.4 & \stSafe & Memout  & \dash   \\
 & 7999 & 15992  & Memout & \dash & 634.4 & \stSafe & Memout  & \dash   \\
 & 9999 & 19992  & Memout & \dash & 1115.6& \stSafe & Memout  & \dash   \\
\midrule

\multirow{9}{*}{Adder}
 & 13   & 44    & 160.0  & \stSafe & 0.1   & \stSafe & 0.6    & \stSafe \\
 & 599  & 2388  & Memout & \dash   & 3.3   & \stSafe & 13.3   & \stSafe \\
 & 999  & 3988  & Memout & \dash   & 8.8   & \stSafe & 38.9   & \stSafe \\
 & 1499 & 5988  & Memout & \dash   & 20.2  & \stSafe & 81.1   & \stSafe \\
 & 1999 & 7988  & Memout & \dash   & 33.1  & \stSafe & 139.3  & \stSafe \\
 & 2999 & 11988 & Memout & \dash   & 71.1  & \stSafe & 444.8  & \stSafe \\
 & 3999 & 15988 & Memout & \dash   & 141.9 & \stSafe & 1025.3 & \stSafe \\
 & 4999 & 19988 & Memout & \dash   & 218.9 & \stSafe & Memout & \dash   \\
 & 5999 & 23988 & Memout & \dash   & 317.8 & \stSafe & Memout & \dash   \\
\midrule

\multirow{7}{*}{Bridge GHZ}
 & 1399 & 2797 & Memout & \dash & 2.4  & \stSafe & 27.3   & \stSafe \\
 & 1799 & 3597 & Memout & \dash & 4.5  & \stSafe & 40.8   & \stSafe \\
 & 1999 & 3997 & Memout & \dash & 5.8  & \stSafe & 67.1   & \stSafe \\
 & 2399 & 4797 & Memout & \dash & 7.6  & \stSafe & 74.9   & \stSafe \\
 & 3599 & 7197 & Memout & \dash & 8.0  & \stSafe & 234.2  & \stSafe \\
 & 4399 & 8797 & Memout & \dash & 10.4 & \stSafe & Memout & \dash   \\
 & 4799 & 9597 & Memout & \dash & 11.0 & \stSafe & Memout & \dash   \\
\midrule

\multirow{3}{*}{\shortstack[l]{Reqomp-\\MCX}}
 & 11 & 9  & 2.8    & \stPhase & 0.1 & \stPhase & 0.3 & \stPhase \\
 & 13 & 11 & 92.4   & \stPhase & 0.1 & \stPhase & 0.3 & \stPhase \\
 & 15 & 13 & Memout & \dash    & 0.1 & \stPhase & 0.3 & \stPhase \\
\midrule

\multirow{3}{*}{\shortstack[l]{Identity\\Random}}
 & 5 & 47 & 0.1 & \stSafe & 0.4 & \stSafe & 0.2 & \stSafe \\
 & 6 & 53 & 0.1 & \stSafe & 1.4 & \stSafe & 0.2 & \stSafe \\
 & 7 & 54 & 0.2 & \stSafe & 5.1 & \stSafe & 0.2 & \stSafe \\
\midrule

\multirow{3}{*}{\shortstack[l]{Random\\Circuit}}
 & 30 & 11 & Memout & \dash  & 0.1 & \stHybrid & 0.3 & \stHybrid \\
 & 50 & 14 & Memout & \dash  & 0.1 & \stHybrid & 0.4 & \stHybrid \\
 & 70 &  4 & Memout & \dash  & 0.1 & \stHybrid & 0.5 & \stHybrid \\
\bottomrule
\end{tabular}
}
\end{table*}
\cref{tab:full_data_unified_final} shows more detail comparing performance of different backends.